\documentclass[format=acmsmall]{acmart}

\setcopyright{acmcopyright}
\acmJournal{TOCL}
\usepackage[utf8]{inputenc} 
\usepackage{pgf}
\usepackage{tikz}
\usetikzlibrary{arrows,automata, positioning}
\usepackage{xspace}

\newcommand{\concept}[1]{{\bf\em\boldmath #1}}

\newcommand{\set}[1]{\left\{ #1 \right\}}

\renewcommand{\epsilon}{\varepsilon}

\newcommand{\A}{\mathcal{A}}
\newcommand{\B}{\mathcal{B}}
\renewcommand{\C}{\mathcal{C}}
\newcommand{\D}{\mathcal{D}}

\newcommand{\N}{\mathbb{N}}
\newcommand{\M}{\mathcal{M}}

\newcommand{\Last}{\text{Last}}
\newcommand{\Run}{\text{Runs}}
\newcommand{\Buchi}{\text{Büchi}}
\newcommand{\CoBuchi}{\text{co-Büchi}}
\newcommand{\Tree}{\text{Trees}}

\newcommand{\initState}{\ensuremath{s_{\init}} }

\newcommand{\tree}{t}

\newcommand{\setstates}{Q}
\renewcommand{\state}{q}
\newcommand{\trans}{\Delta}
\newcommand{\init}{\textit{in}}

\newcommand{\tpath}{\lambda}
\newcommand{\plays}{\textrm{Plays}}

\newcommand{\thepath}{\mathsf{path}}
\newcommand{\Paths}{\mathsf{Paths}}
\newcommand{\forallone}{\forall^{=1}}
\newcommand{\forallonep}{\forall^{=1}_{\!\thepath}}

\newcommand{\MSOorStS}{MSO}
\newcommand{\MSO}{\MSOorStS\xspace}
\newcommand{\MSOzero}{\MSOorStS+$\forallone$\xspace}
\newcommand{\MSOzeropath}{\MSOorStS+$\forallonep$\xspace}
\newcommand{\CTLs}{\textsf{\ensuremath{\textsf{\textup{CTL}}^*}}\xspace}

\DeclareMathOperator{\suc}{succ}
\newcommand*{\pipe}{\ |\ }
\newcommand*{\words}{\{0,1\}^*}

\newcommand{\arena}{\mathcal{G}}

\newcommand{\vini}{v_\init}

\newcommand{\Eloise}{Éloïse\xspace}
\newcommand{\Abelard}{Abélard\xspace}
\newcommand{\Random}{Random\xspace}
\newcommand{\Eloisei}{E}
\newcommand{\Abelardi}{A}
\newcommand{\Ei}{E}
\newcommand{\Ai}{A}
\newcommand{\VE}{V_\Eloisei}
\newcommand{\VA}{V_\Abelardi}
\newcommand{\VR}{V_R}
\newcommand{\game}{\mathbb{G}}

\newcommand{\WC}{\Omega}
\newcommand{\Plays}{\mathrm{Plays}}

\newcommand{\play}{\lambda}
\renewcommand{\phi}{\varphi}

\newcommand{\strat}{\sigma}
\newcommand{\sini}{v_{\mathrm{in}}}
\newcommand{\states}{V}

\newcommand{\proba}{\mathrm{Pr}}

\newcommand{\rank}{\mathrm{rank}}

\begin{document}

\theoremstyle{acmdefinition}
\newtheorem{remark}[theorem]{Remark}

\title{Alternating Tree Automata with Qualitative Semantics}

\author{Rapha\"el Berthon}
\affiliation{%
  \institution{Universit\'e libre de Bruxelles}
  \city{Bruxelles}
  \country{Belgium}}
\email{Raphael.Berthon@ulb.ac.be}

\author{Nathana\"el Fijalkow}
\affiliation{%
  \institution{CNRS \& LaBRI}
  \city{Bordeaux}
  \country{France}}
\email{Nathanael.Fijalkow@labri.fr}

\author{Emmanuel Filiot}
\affiliation{%
  \institution{Universit\'e libre de Bruxelles}
  \city{Bruxelles}
  \country{Belgium}}
\email{efiliot@ulb.ac.be}

\author{Shibashis Guha}
\affiliation{%
  \institution{Universit\'e libre de Bruxelles}
  \city{Bruxelles}
  \country{Belgium}}
\email{Shibashis.Guha@ulb.ac.be}

\author{Bastien Maubert}
\orcid{0000-0002-9081-2920}
\affiliation{%
  \institution{Universit\`a degli Studi di Napoli ``Federico II''}
  \department{DIETI}
  \city{Naples}
  \country{Italy}}
\email{bastien.maubert@gmail.com}

\author{Aniello Murano}
\affiliation{%
  \institution{Universit\`a degli Studi di Napoli ``Federico II''}
  \department{DIETI}
  \city{Naples}
  \country{Italy}}
\email{murano@na.infn.it}

\author{Laureline Pinault}
\affiliation{%
  \institution{Univ Lyon, CNRS, ENS de Lyon, UCB Lyon 1, LIP}
  \city{Lyon}
  \country{France}}
\email{Laureline.Pinault@ens-lyon.fr}

\author{Sophie Pinchinat}
\affiliation{%
  \institution{Univ Rennes, CNRS, IRISA}
  \city{Rennes}
  \country{France}}
\email{Sophie.Pinchinat@irisa.fr}

\author{Sasha Rubin}
\affiliation{%
  \institution{University of Sydney}
  \city{Sydney}
  \country{Australia}}
\email{rubin@forsyte.at}

\author{Olivier Serre}
\orcid{0000-0001-5936-240X}
\affiliation{%
  \institution{Université de Paris, IRIF, CNRS}
  \streetaddress{Bâtiment Sophie Germain, Case courrier 7014, 
8 Place Aurélie Nemours}
  \city{Paris Cedex 13}
  \postcode{75205}
  \country{France}}
\email{Olivier.Serre@cnrs.fr}

\begin{abstract}
We study alternating automata with qualitative semantics over infinite binary trees: alternation means that two opposing players construct a decoration of the input tree called a run, and the qualitative semantics says that a run of the automaton is accepting if almost all branches of the run are accepting. In this paper we prove a positive and a negative result for the emptiness problem of alternating automata with qualitative semantics.

The positive result is the decidability of the emptiness problem for the case of Büchi acceptance condition. An interesting aspect of our approach is that we do not extend the classical solution for solving the emptiness problem of alternating automata, which first constructs an equivalent non-deterministic automaton. Instead, we directly construct an emptiness game making use of imperfect information.

The negative result is the undecidability of the emptiness problem for the case of co-Büchi acceptance condition. This result has two direct consequences: the undecidability of monadic second-order logic extended with the qualitative path-measure quantifier, and the undecidability of the emptiness problem for alternating tree automata with non-zero semantics, a recently introduced probabilistic model of alternating tree automata.
\end{abstract}

\acmSubmissionID{}


\begin{CCSXML}
<ccs2012>
<concept>
<concept_id>10003752.10003766.10003770</concept_id>
<concept_desc>Theory of computation~Automata over infinite objects</concept_desc>
<concept_significance>500</concept_significance>
</concept>
<concept>
<concept_id>10003752.10003766.10003772</concept_id>
<concept_desc>Theory of computation~Tree languages</concept_desc>
<concept_significance>500</concept_significance>
</concept>
<concept>
<concept_id>10003752.10003766.10003773.10003775</concept_id>
<concept_desc>Theory of computation~Quantitative automata</concept_desc>
<concept_significance>500</concept_significance>
</concept>
<concept>
<concept_id>10003752.10003790</concept_id>
<concept_desc>Theory of computation~Logic</concept_desc>
<concept_significance>500</concept_significance>
</concept>
<concept>
<concept_id>10003752.10003753.10003757</concept_id>
<concept_desc>Theory of computation~Probabilistic computation</concept_desc>
<concept_significance>500</concept_significance>
</concept>
</ccs2012>
\end{CCSXML}

\ccsdesc[500]{Theory of computation~Automata over infinite objects}
\ccsdesc[500]{Theory of computation~Tree languages}
\ccsdesc[500]{Theory of computation~Quantitative automata}
\ccsdesc[500]{Theory of computation~Logic}
\ccsdesc[500]{Theory of computation~Probabilistic computation}


\keywords{tree automata, $\omega$-regular conditions, almost-sure semantics}

\maketitle

\renewcommand{\shortauthors}{R. Berthon et al.}

\section{Introduction}
The study of tree-automata models can be organised by distinguishing three semantic features.

  The first feature is the \textit{operational mode}:
  deterministic, non-deterministic, universal, probabilistic, and
  alternating, are the most studied notions. Intuitively, in each case, an
  automaton reading an input tree (with labels on the nodes)
  constructs a decoration of this tree called a run, which is itself a
  tree. The run labels nodes of the tree by states respecting the
  local constraints imposed by the transition relation of the
  automaton. In the deterministic case, a state and a letter uniquely
  determine the labels at the level below in the run, hence there is a
  unique run. In the non-deterministic, universal, and alternating
  case, there may be several valid transitions at each node, yielding
  possibly several runs on a single tree. In the non-deterministic
  case we say that the tree is accepted if there exists an accepting
  run, \textit{i.e.} the choices are existential. In the universal
  case, we say that the tree is accepted if all runs are accepting,
  \textit{i.e.} the choices are universal. The alternating case
  unifies both previous cases by introducing existential and universal
  transitions.

The second feature is the \textit{branching semantics}. The classical one says that a run is accepting if all its branches satisfy a given acceptance condition. We are concerned in this paper with the qualitative semantics, which is an alternative branching semantics introduced by Carayol, Haddad, and Serre~\cite{CHS14}. The qualitative semantics says that a run is accepting if almost all its branches satisfy a given acceptance condition, in other words if by picking a branch uniformly at random it almost-surely satisfies the condition. The paper~\cite{CHS14} showed that non-deterministic and probabilistic tree automata with qualitative semantics are both robust computational models with appealing algorithmic properties.

The third feature is the \textit{acceptance condition} (on branches), with $\omega$-regular conditions such as B{\"u}chi and parity conditions being the most important for their tight connections to logical formalisms; see, e.g., \cite{Thomas97}.

\vskip1em
One motivation for studying tree automata with qualitative semantics is to extend the deep connections between automata and monadic second-order logic (MSO) which hold for the classical semantics~\cite{Rabin:TAMS69}.
Indeed, the general goal is to construct decidable extensions of MSO over infinite trees; we review some of the efforts and results obtained in this direction. A (unary) generalised quantifier is of the form ``the set of all sets $X$ that satisfy $\varphi$ has the property $C$'', where $C$ is a property of sets. For instance, the ordinary existential quantifier $\exists X. \varphi$ corresponds to the property $C$ of being a non-empty set. More interestingly, the quantifier ``there exist infinitely many $X$ such that $\varphi$'' corresponds to the property $C$ of being infinite. It turns out that certain cardinality quantifiers such as ``there exist infinitely many $X$'' and ``there exist continuum many $X$'' do not add expressive power to \MSO over the infinite binary tree (in fact, they can be effectively eliminated)~\cite{Rabinovich:FI10}.
On the other hand, adding the generalised quantifier ``the set of all sets $X$ satisfying $\varphi$ has Lebesgue-measure one'' results in an undecidable theory~\cite{michalewski2018monadic}. A weaker version of this quantifier is ``the set of paths of the tree that satisfy $\varphi$ has Lebesgue-measure one''. Intuitively, this quantifier, written $\forallonep$,  means that a random path almost-surely satisfies $\varphi$, where a random path is generated by repeatedly flipping a coin to decide whether to go left or right. It was proved in~\cite{Bojanczyk16,BGK17} that adding the quantifier $\forallonep$ to a restriction of \MSO called ``thin MSO'' yields a decidable logic, but the decidability of \MSOzeropath was left open in~\cite{michalewski2016measure,michalewski2018monadic}.
The emptiness problem for non-deterministic  parity tree automata with
qualitative semantics can easily be expressed using \MSOzeropath, as
already observed in~\cite{michalewski2018monadic}, and this is also the case for universal tree automata with qualitative semantics.

\smallskip
In this paper, we initiate the study of alternating automata with qualitative semantics, and focus on the emptiness problem.
We present {a positive result and a negative result} that delimit a clear and sharp decidability frontier.

\subsection*{Contributions}
%----------------------------
{The positive result is the decidability of the emptiness problem for the case of the B{\"u}chi acceptance condition (Theorem~\ref{thm:decidability-altBuchi}).} 

The usual roadmap for solving the emptiness problem for alternating automata is to first construct an equivalent non-deterministic automaton, and then to construct an emptiness game for the non-deterministic automaton, i.e., a game such that the first player wins if and only if the automaton is non-empty. This first step is an effective construction of an equivalent non-deterministic automaton, which in some cases is not possible, unknown, or computationally too expensive. In the case at hand the second situation arises: we do not know whether alternating automata with qualitative semantics can effectively be turned into equivalent non-deterministic ones. {We remark that our undecidability result shows that there is no such effective construction for co-B{\"u}chi conditions (but there might be one for the B{\"u}chi conditions).}

{Here, instead,} we develop a new approach which directly constructs an emptiness game for the alternating automaton. The emptiness game we construct uses \emph{imperfect information}. Our construction extends the notion of \emph{blindfold games} {of} Reif~\cite{Reif79}, used to check universality of non-deterministic automata over finite words. The key ingredient to proving the correctness of our imperfect information emptiness game is a new positionality result for stochastic B\"uchi games on {certain \emph{infinite} arenas (that we call \emph{chronological})}. To the best of our knowledge, very few positionality results are known in the literature that combine both {stochastic features and infinite arenas}; { a notable exception is~\cite{Kucera_foxes}.}

\smallskip
The negative result is the undecidability of the emptiness problem for the case of the co-B\"uchi acceptance condition. In fact, our main technical contribution (Theorem~\ref{thm:undecidability-univcoB}) is to establish the undecidability already for universal automata (a special subclass of alternating automata).

We establish this by a chain of reductions that consider various classes of automata (both on infinite words and trees). We initially resort to the known undecidability of the value 1 problem for probabilistic automata on finite words~\cite{GO10} to deduce the undecidability of the emptiness problem for \emph{simple} probabilistic co-B{\"u}chi automata on infinite words (Proposition~\ref{prop-simple-cobuchi}). Here, \emph{simple} means that the transitions of the automaton only involve probabilities in $\{0,\frac{1}{2},1\}$. Then, we reduce the latter problem to the original emptiness problem for universal co-B{\"u}chi tree automata with qualitative semantics, hence proving our negative result. The correctness of this last reduction relies on particular properties of another class of automata, namely, \emph{probabilistic tree automata}.

Our negative result has two interesting consequences: the undecidability of \MSOzeropath, and of the emptiness problem for alternating tree automata with non-zero semantics, a model combining sure, almost-sure, and positive semantics and studied in~\cite{fournier2018alternating}.

\subsection*{Related Work}
%----------------------------
The study of automata with qualitative semantics was initiated in~\cite{CHS14} with several decidability results. The first result is a polynomial-time algorithm {entailing} the decidability of the emptiness problem for non-deterministic parity tree automata with qualitative semantics~\cite{CHS14}, obtained through a polynomial reduction to the almost-sure problem for Markov decision processes (for which a polynomial-time algorithm is known from~\cite{CourcoubetisY90}).  {This} reduction extends to probabilistic tree automata with qualitative semantics, showing an equivalence with partial-observation Markov decision processes. {It} is then used to prove the decidability of the emptiness problem for probabilistic B{\"u}chi tree automata with qualitative semantics~\cite{CHS14}.

Alternation was later considered by Fijalkow, Pinchinat, and Serre in~\cite{fijalkow2013emptiness} where the focus was on designing a {novel emptiness checking} procedure working directly on alternating automata, \textit{i.e.} directly building an emptiness \emph{imperfect-information} game without {making use of the} intermediate transformation to a non-deterministic automaton: this was successfully applied to classical alternating parity tree automata as well as to alternating B{\"u}chi tree automata with qualitative semantics ({see Theorem~\ref{thm:decidability-altBuchi}}).

This line of work was pursued using the related model of non-zero automata. The first decidability result was obtained for the subclass of zero automata~\cite{BGK17}, yielding the decidability of the thin restriction of \MSOzeropath. A second decidability result concerned the class of alternating zero automata~\cite{fournier2018alternating}, restricting the abilities of the second player. This {latter} result is applied to solve the satisfiability problem of a probabilistic extension of \CTLs. The general case {of non-zero automata} was left open. We close it {negatively {(thanks to our negative result)} since alternating tree automata with non-zero semantics subsume universal tree automata with qualitative semantics}.

A side result in~\cite{fijalkow2013emptiness} states the undecidability {of the emptiness problem for} alternating co-B{\"u}chi automata with qualitative semantics.
The proof, not given in the conference proceedings, is rather sketchy in the full version~\cite{fijalkow2013emptinessComplete}. The proof we give here (Theorem~\ref{thm:undecidability-univcoB}) follows the same lines but clarifies a technical loophole in the original proof. Indeed, the last reduction requires the undecidability of the emptiness problem for probabilistic co-B{\"u}chi \textit{simple} automata over infinite words, where \emph{simple} means that the transitions probabilities are either $0$, $\frac{1}{2}$, or $1$.
The undecidability result was known only for general automata, {while we refine it for the simple ones, thus filling in the gap of the undecidability proof in the full version~\cite{fijalkow2013emptinessComplete}}. 

More recently, Berthon \emph{et al.}~\cite{Berthonetal19} proved {the}  slightly weaker undecidability result {that} emptiness is undecidable for universal \emph{parity} tree automata with qualitative semantics. {Although} their proof follows the same lines as~\cite{fijalkow2013emptinessComplete}, {the result is weaker because they need a stronger acceptance condition to obtain simple automata and prove the correctness of the original reduction}. {Still,} their result is strong enough to {entail} the undecidability of \MSOzeropath, the main contribution of their work.

There is another proof of the undecidability of \MSOzeropath, obtained independently and {at the same time} {as~\cite{Berthonetal19}} by Boja{\'n}czyk, Kelmendi, and Skrzypczak~\cite{BKS19}. Their proof technique is very different from ours: they obtain undecidability by a direct encoding of two-counter machines into the logic. However, the core technical part of the paper is not the reduction from counter machines (which is nevertheless tricky), but a {crucial technical lemma} used to encode runs of counter machines and to prove the correctness of the reduction\footnote{More precisely, the lemma states that for a set $D$ of pairwise disjoint finite paths in the infinite binary tree called \emph{intervals}, there is an \MSOzeropath formula that, when true at the root of the infinite binary tree, is equivalent to having with probability $1$, a branch $\pi$ and some integer $\ell$ such that with finitely many exceptions, if an interval intersects $\pi$ then it is of length $\ell$.}. {The proof of this lemma is involved: it} mostly relies on tools (such as asymptotic behaviours of vector sequences) previously used to show undecidability of MSO+U logic over infinite words. {We remark that  \MSOzeropath is known as} MSO+$\nabla$ in~\cite{BKS19}.

\subsection*{Organisation of the Paper}
%----------------------------
Section~\ref{sec-prelim} presents the different classes of automata used for our main undecidability result, 
relying on Markov chains as a unifying notion {to define} acceptance by these different automata. 
Section~\ref{sec-dec-Buchi} gives our decidability result for alternating B{\"u}chi tree automata. Section~\ref{sec-undec-automata} is about our undecidability result {for universal co-Büchi tree automata}, while Section~\ref{sec-cor} presents its consequences for \MSOzeropath (Section~\ref{sec:msozero}) and for alternating automata with non-zero semantics (Section~\ref{sec:nonz}).

\section{Preliminaries}
\label{sec-prelim}
Throughout the paper we implicitly fix a finite alphabet $\Sigma$.
We denote by $\Sigma^*$ the set of \concept{finite words} over
$\Sigma$ and by $\Sigma^\omega$ the set of \concept{infinite words}
over $\Sigma$. {We let $\epsilon$ denote the empty word, and for a word
$u\in \Sigma^*$, $|u|$ denotes its length. Finally, we write 
  $\Sigma^k$ for the set of words over $\Sigma$ of length $k$.}

{The \concept{infinite binary tree} is $\set{0,1}^*$, elements of $\{0,1\}^*$ are called its \concept{nodes}, and elements of $\{0,1\}^\omega$ are called its (infinite) \concept{branches}}.  For a finite alphabet
$\Sigma$, a \concept{$\Sigma$-tree} is a function
$t : \set{0,1}^* \to \Sigma$ and we write $\Tree(\Sigma)$ for the set
of $\Sigma$-trees. For a branch $b=b_1b_2\cdots \in \set{0,1}^\omega$
we denote by {$t[b]=t(\epsilon)t(b_1)t(b_1b_2)t(b_1b_2b_3)\cdots \in \Sigma^\omega$ the
infinite word read in $t$ along the branch $b$.}

A \concept{distribution} over a set $Q$ is a function
$\delta : Q \rightarrow [0,1]$ such that
$\sum_{q \in Q} \delta(q) = 1$.  Any distribution $\delta$ considered
in the paper is implicitly assumed to have a finite support,
i.e. $\{q\in Q\mid \delta(q)\neq 0\}$ is finite.  
{For $Q'\subseteq Q$, we write $\sum_{q\in Q'}p_q \cdot q$ for the
distribution that assigns probability $p_q$ to $q \in Q'$ and $0$ to $q\in Q\setminus Q'$. For example, 
$\frac{1}{2} q_1 + \frac{1}{2} q_2$ is the distribution $\delta$ such
that $\delta(q_1) = \delta(q_2) = \frac{1}{2}$, unless $q_1=q_2$ in which case $\delta(q_1) = \delta(q_2) =1$, and $\delta(q)=0$ for every other element $q$.}  The set of distributions over $Q$ is denoted
$\D(Q)$.

A \concept{Markov chain} $\M=(S,\initState,T)$ is given by a
\emph{possibly infinite} set of states $S$, an initial state
$\initState \in S$, and a probability transition function $T : S \to \D(S)$. An
(infinite) \concept{path} in $\M$ is an infinite sequence of states 
$s_0s_1s_2\ldots\in S^\omega$ such that $s_0=\initState$ and
$T(s_i)(s_{i+1})>0$ for every $i\geq 0$. A \concept{cone} is a set of paths of
the form $u\cdot S^\omega$ for some $u\in S^*$. Now, consider the
$\sigma$-algebra over paths in $\M$ built from the set of cones. Then,
a classical way to equip this $\sigma$-algebra with a probability
measure $P$ is to recursively define it on the set of cones as follows:
$$P(s_0s_1\cdots s_k \cdot S^\omega)=
		\begin{cases}
			1 & \text{if $k=0$}\\			P(s_0\cdots s_{k-1} \cdot S^\omega)\cdot T(s_{k-1})(s_k) & \text{otherwise}\\	
		\end{cases}
	$$ 
and then to extend it (uniquely) to the $\sigma$-algebra thanks to Carath{\'e}odory's extension theorem (we refer the reader to Reference~\cite{Puterman94} for more details on this classical construction).

When needed, for a given length $k$, we also see $P$ as a probability measure on paths of length $k$ (i.e. elements in $S^k$) by defining  the probability measure of $u\in S^k$ as the probability of the cone $u\cdot S^\omega$.

\subsection{Two-Player Perfect-Information Stochastic Games}\label{section:perfect_Info_Games}
%--------------------------------------------------------------------------------------

A \concept{graph} is a pair $G = (V,E)$ where $V$ is a (possibly infinite) set of \concept{vertices} and $E\subseteq V\times V$ is a set of \concept{edges}.
{For every vertex $v$, let $E(v) = \set{w \mid (v,w) \in E}$, and say that $v$ is a \concept{dead-end} if $E(v) = \emptyset$.
In the rest of the paper, we only consider graphs of finite out-degree, i.e. such that $|E(v)|$ is finite for every vertex $v\in V$, and without dead-ends.}

A (turn-based) \concept{stochastic arena} is a tuple $\arena = (G,\VE,\VA,\VR,\delta,\vini)$ 
where $G = (V,E)$ is a graph, $(\VE,\VA,\VR)$ is a partition of the vertices among two players, 
\Eloise and \Abelard, and an extra player \Random, $\delta:\VR\rightarrow \D(V)$ is a map 
such that for all $v \in \VR$ the support of $\delta(v)$ is included in $E(v)$, and $\vini \in V$ is an \concept{initial} vertex.
In a vertex $v \in \VE$ (\emph{resp.} $v \in \VA$) \Eloise (\emph{resp.} \Abelard) chooses a successor vertex from $E(v)$, 
and in a random vertex $v \in \VR$, a successor vertex is chosen according to the probability distribution $\delta(v)$. A
\concept{play} $\play=v_0v_1v_2\cdots$ is an infinite sequence of vertices starting
from the initial vertex, i.e. $v_0=\vini$, and  {such that, for every $k\geq 0$, $v_{k+1}\in E(v_k)$ if $v_k\in\VE\cup\VA$ and $\delta(v_k)(v_{k+1})>0$ if $v_k\in\VR$}. A \concept{history} is a finite prefix of a play.

A (pure\footnote{We only consider pure strategies, as these are sufficient for our purpose. However, our
  main results on positionality
  (Theorems~\ref{theo:positionalityChrono} and
  \ref{theo:positionalityChronoReach}) remain true for
    randomised strategies as later discussed in
  Remark~\ref{rk:purevsrandomised}.}) \concept{strategy} for \Eloise
is a function $\strat_\Ei: V^*\cdot\VE \rightarrow V$ such that for
every history $\lambda\cdot v \in V^*\cdot\VE$ one has
$\strat_\Ei({\lambda\cdot} v) \in E(v)$.  Strategies of \Abelard are defined
likewise, and usually denoted $\strat_\Ai$.

A play $\tpath = v_0v_1v_2\ldots$ is \concept{consistent} with a pair
of strategies $(\sigma_\Ei,\sigma_\Ai)$ for \Eloise and \Abelard if
the players always choose their move according to their
strategy. Formally, for all $k\geq 0$ the following should hold: if
$v_k$ is controlled by \Eloise then
$v_{k+1}=\sigma_\Ei({v_0\ldots}v_k)$ and if it is controlled
by \Abelard then $v_{k+1}=\sigma_\Ai({v_0\ldots}v_k)$. {The set of plays consistent with
$(\sigma_\Ei,\sigma_\Ai)$ is denoted
$\plays^{\arena}_{\sigma_\Ei,\sigma_\Ai}$, and a history is consistent with
$(\sigma_\Ei,\sigma_\Ai)$ if it is the finite prefix of some play in $\plays^{\arena}_{\sigma_\Ei,\sigma_\Ai}$}.

{In order to equip the set $\plays^{\arena}_{\sigma_\Ei,\sigma_\Ai}$ with a probability measure, {we define the following} Markov chain
$\mathcal{M}^{\arena}_{\sigma_\Ei,\sigma_\Ai}$: its set of states is the set of {histories consistent with $(\sigma_\Ei,\sigma_\Ai)$}, its initial state is $\vini$, and its probability transition function $T$ is defined
 %(according to our notation for distributions) 
 by}
{
$$
T(\lambda\cdot v) = 
\begin{cases}
\lambda\cdot v\cdot \sigma_\Ei(\lambda\cdot v) & \text{if } v \in V_\Ei\\
\lambda\cdot v\cdot \sigma_\Ai(\lambda\cdot v) & \text{if } v \in V_\Ai\\
\sum_{v'\in E(v)}\delta(v)(v')\ \lambda\cdot v\cdot v' & \text{if } v \in V_R
\end{cases}$$}

Then, the set $\plays^{\arena}_{\sigma_\Ei,\sigma_\Ai}$ of those plays consistent with $(\sigma_\Ei,\sigma_\Ai)$ is {in bijection with} the set of infinite paths in the Markov chain $\mathcal{M}^{\arena}_{\sigma_\Ei,\sigma_\Ai}$. Hence, the associated probability measure $P^{\arena}_{\sigma_\Ei,\sigma_\Ai}$ can be used as a probability measure for measurable subsets of $\plays^{\arena}_{\sigma_\Ei,\sigma_\Ai}$. \\
{When $\arena$ is understood, we omit it and simply write $P_{\sigma_\Ei,\sigma_\Ai}$ and $\plays_{\sigma_\Ei,\sigma_\Ai}$.}

A \concept{winning condition} is a subset\footnote{Formally, one needs to require that $\WC$ is measurable for the probability measure $P_{\sigma_\Ei,\sigma_\Ai}$, which is always trivially true in this paper.} 
$\WC\subseteq V^\omega$ and a (two-player perfect-information) \concept{stochastic game} is a pair 
$\game = (\arena,\WC)$.

A strategy $\sigma_\Ei$ for \Eloise is \concept{surely winning} if $\Plays_{\sigma_\Ei,\sigma_\Ai} \subseteq \WC$ for every strategy $\sigma_\Ai$ of \Abelard; 
it is  \concept{almost-surely winning} if $P_{\sigma_\Ei,\sigma_\Ai}(\WC) = 1$
for every strategy $\sigma_\Ai$ of \Abelard. 
Similar notions for \Abelard are defined dually.
{\Eloise \concept{ surely} (resp.\ \concept{almost-surely}) \concept{wins} if she has a \concept{surely} (resp.\ \concept{almost-surely}) winning strategy.}

A \concept{reachability game} is a {stochastic game whose winning condition is } of the
form $V^*FV^\omega$ for some subset $F\subseteq V$, i.e. winning plays
are those that eventually visit a vertex in $F$.  A \concept{Büchi
  game} is a {stochastic game whose winning condition is } of the form
$\bigcap_{i\geq 0}V^iV^*FV^\omega$ for some subset $F\subseteq V$,
i.e. winning plays are those that infinitely often visit a vertex in
$F$.  Finally, a \concept{co-Büchi game} is {stochastic game whose winning condition is }of the form $V^*(V\setminus F)^\omega$ for some subset
$F\subseteq V$, i.e. winning plays are those that finitely often visit
a vertex in $F$. When it is clear from the context, we write
$\game=(\arena,F)$ (i.e. write $F$ instead of $\WC$) for the
reachability (\emph{resp.} Büchi, co-Büchi) game {relying on} $F$.

A \concept{positional strategy} $\strat$ is {a strategy} that does not require {any}
memory, i.e.\ such that for any two histories of the form
$\play \cdot v$ and $\play' \cdot v$, one has
$\strat(\lambda \cdot v) = \strat(\lambda' \cdot v)$. 
{Positional strategies only depend on the current vertex, and for convenience
they are written as functions from $V$ into $V$.}

A game is \emph{deterministic} whenever $\VR = \emptyset$. It is well-known (see e.g.~\cite{Zie98}) that positional strategies suffice to surely win in deterministic games with {\emph{a parity winning condition}, which we do not define but captures the reachability, Büchi, and co-Büchi winning conditions that we are interested in.}

\begin{theorem}[Positional determinacy ~\cite{Zie98}]\label{theo:posDet}
Let $\game$ be a deterministic parity game. Then, either \Eloise or \Abelard has a positional surely winning strategy.
\end{theorem}

For stochastic games, the following result is well-known (see e.g.~\cite{GZ07} for a slightly more general result).

\begin{theorem}\label{theo:posDetSt}
Let $\game$ be a stochastic parity game played on a \emph{finite} arena. 
If \Eloise almost-surely wins then she {has an a positional almost-surely winning strategy}.
\end{theorem}

Note that dropping the assumption that the arena is finite {substantially}
changes the situation. Indeed, for infinite arenas, even with a
reachability condition and assuming finite out-degree, almost-surely
winning strategies for \Eloise may require infinite memory
\cite[Proposition~5.7]{Kucera_foxes}.  However, {imposing a natural structural
restriction on the (possibly infinite) arena, namely to be chronological, yields a
result like Theorem~\ref{theo:posDetSt} for Büchi games, see Theorem~\ref{theo:positionalityChrono}.}

\subsection{Two-Player Imperfect-Information Stochastic Büchi Games}
%----------------------------------------------------
\label{sec:imperfect-info}

We now introduce a subclass of the usual games with
imperfect information which is essentially a stochastic version of the model in~\cite{CDHR07}. Our model of imperfect-information games is quite restrictive compared to general models developed in~\cite{GS09,BGG09,DBLP:journals/tocl/Chatterjee014,CarayolLS18}, as in our setting \Abelard is perfectly informed. 
However, it turns out to be expressive enough to be used as a central tool to check emptiness for alternating Büchi tree automata with qualitative semantics.

A \concept{{stochastic} arena of imperfect information} is a tuple $\arena =
(\states,A,T,\sim,v_\init)$ where $\states$ is a \emph{finite} set of
{vertices}, $v_\init \in \states$ is an initial {vertex}, $A$ is the finite
alphabet of \Eloise's actions, $T \subseteq \states \times A \times
\D(\states)$ is a stochastic transition relation and $\sim$
is an equivalence relation over $\states$ {that denotes the observational capabilities of \Eloise  and therefore imposes restrictions on legitimate strategies for her (see further)}.  We additionally require
that for all $(v,a) \in \states \times A$ there is at least one
$\delta \in \D(\states)$ such that $(v,a,\delta) \in T$. 

A \concept{play} starts from the initial vertex $v_\init$
and proceeds as follows: \Eloise plays an action $a_0 \in A$, then
\Abelard resolves the non-determinism by choosing a distribution
$\delta_0$ such that $(v_\init,a_0,\delta_0) \in T$ and finally a new
vertex is randomly chosen according to $\delta_0$.  Then,
\Eloise plays a new action, \Abelard resolves the non-determinism and
a new vertex is randomly chosen, and so on forever.  Hence,
a play is an infinite word
$v_\init a_0 \delta_0 v_1 a_1 \delta_1 v_2\cdots \in (\states \cdot A
\cdot \D(\states))^\omega$ . A \concept{history} is a prefix of a
play ending in a vertex in $\states$.

An \concept{imperfect-information stochastic Büchi game} is a pair
$\game = (\arena,F)$ where $\arena$ is a {stochastic} arena
of imperfect information with a subset of states $F\subset \states$
used to define the Büchi winning condition {as follows:} a
play
$\play=v_{0} {a_0 \delta_0}v_1{a_1 \delta_1} v_2
\cdots$ in $\game$ is won by \Eloise if, and only if, the set
$\{i\geq 0\mid v_i\in F\}$ is infinite, i.e. winning plays are those
that infinitely often visit a vertex in $F$.

{The imperfect-information of the game is modelled by the equivalence relation
$\sim$ that conveys which vertices \Eloise cannot distinguish, namely those that are $\sim$-equivalent. {We will write 
$\states_{/_\sim}$ for the set of equivalence classes of $\sim$ in $\states$,
and for every $v \in \states$, we will write $[v]_\sim$ for its
$\sim$-equivalence class.}

Relation $\sim$ plays a crucial role when defining
strategies for \Eloise. Intuitively, \Eloise should not play differently
in two indistinguishable plays, where the indistinguishability of
\Eloise is based on \emph{perfect recall}~\cite{FaginHMV95}:}
\Eloise cannot distinguish two histories
$v_\init a_0 \delta_0 v_1 a_1 \delta_1 \cdots v_\ell$ and
$v_\init' a_0' \delta_0' v_1' a_1' \delta_1' \cdots v_\ell'$ {whenever} $v_i \sim v_i'$ for all $i \leq \ell$ and $a_i = a_i'$ for all
$i < \ell$. {Note that in particular, \Eloise does not observe \Abelard's choices for the distributions along a play.}  Hence, a (pure\footnote{Again, as for perfect information
  games, we do not consider randomised strategies as pure strategies
  are the right model for our purpose.}) \concept{strategy} for
\Eloise is a function
$\strat_\Ei : (\states_{/_\sim} \cdot  A)^* \cdot (\states_{/_\sim}) \rightarrow A$
assigning an action to every set of indistinguishable histories.  \Eloise \concept{ respects a strategy}
$\strat_\Ei$ during a play
$\play = v_\init a_0 \delta_0 v_1 a_1 \delta_1 \cdots$ if
$a_{i+1} = \strat_\Ei([v_\init]_\sim a_0 [v_1]_\sim \cdots[v_i]_\sim)$,
for all $i \geq 0$.

A strategy for \Abelard is defined as a function $\strat_\Ai:(V\cdot A\cdot \D(V))^*(V\cdot A)\rightarrow \D(V)$ such that $(v,a,\strat_\Ai(\lambda\cdot v\cdot a ))\in T$ for every $\lambda\in (V\cdot A\cdot \D(V))^*$. \Abelard \concept{ respects a strategy} $\strat_\Ai$ during a play
$\play = v_\init a_0 \delta_0 v_1 a_1 \delta_1 \cdots$ if
$\delta_{i} = \strat_\Ai(v_\init a_0 \delta_0 v_1 a_1 \delta_1 \cdots v_i a_i)$,
for all $i \geq 0$.

Exactly as in the perfect-information setting, one associates with a pair of strategies $(\strat_\Ei,\strat_\Ai)$ the set $\plays^{\arena}_{\strat_\Ei,\strat_\Ai}$ of those plays where \Eloise (\emph{resp.} \Abelard) respects $\strat_\Ei$ (\emph{resp.} $\strat_\Ai$), and equip it with a probability measure.

Finally, a strategy $\strat_\Ei$ for \Eloise is
\concept{almost-surely winning} if, against any strategy $\strat_\Ai$ for \Abelard, the set of winning plays for \Eloise has measure $1$ for the probability measure on $\plays^{\arena}_{\strat_\Ei,\strat_\Ai}$.

\begin{remark}
It is important to note that \Eloise may not observe whether a vertex belongs to $F$ 
as we do not require that $v \sim v' \Rightarrow (v\in F\Leftrightarrow v'\in F)$. 
In particular, this has to be taken into account when eventually solving the game. 
\end{remark}

The following decidability {result} will be crucial in Section~\ref{sec:emptiness-checking}.

\begin{theorem}[\cite{DBLP:journals/tocl/Chatterjee014,CarayolLS18}]\label{theo:AS_impInfo}
	Let $\game$ be an  {imperfect-information stochastic} Büchi game. One can decide in exponential time whether \Eloise has an almost-surely winning strategy in $\game$.
\end{theorem}

\subsection{Probabilistic Automata on Finite Words} 
%------------------------------------------------
Probabilistic automata on finite words {generalize non-deterministic automata by letting the transition function map} a state and a letter to a distribution over states~\cite{Rabin:IC63}.
The reference book for early developments on probabilistic automata is due to Paz~\cite{paz2014introduction}.

A \concept{probabilistic word automaton} is a tuple $\A = (Q, q_\init, \delta)$, where $Q$ is the finite set of states,
$q_\init$ is the initial state, and $\delta : Q \times \Sigma \to \D(Q)$ is the transition function. 
We say that a probabilistic automaton is \concept{simple} when the {distribution} $\delta(q,a)$ is always of the form $\frac{1}{2} q_1 + \frac{1}{2} q_2$ (possibly with $q_1=q_2$).

Intuitively, a finite word $u=u_1 \dots u_{k} \in \Sigma^*$ induces a set of runs of $\A$ each of which comes with a probability of being realised; if one fixes a set of final states, the \concept{acceptance probability} of $u$ by $\A$ is {the mere} sum of the probabilities of those runs of $\A$ over $u$ that end in a final state. To formally define {acceptance probability} (and {extend it further} to richer settings) we associate with $\A$ and $u$ a Markov chain $\M_\A^u$ as follows.

{The Markov chain $\M_\A^u$ has the (finite) set of states 
$Q \times \set{0,\dots,k}$, the initial state $(q_\init,0)$, and the}
 probability transition function $T_\A^u$ defined for every $(p,i)\in Q\times
 \set{0,\dots,k-1}$  (we do not define it for states of the form $(p,k)$ that will be useless) by
\[T_\A^u((p,i)) = \sum_{q \in Q} \delta(p,u_i)(q) \cdot (q,i+1)\]

Call a finite path of length $k+1$ of $\M_\A^u$ a \concept{run} of
$\A$ on $u$ and let $P_\A^u$ be the probability measure on runs induced
by $\M_\A^u$. Given a subset of (final) states $F \subseteq Q$, call
$\Last(F)$ the set of runs whose (first coordinate of the) last state
is in $F$. We then define the acceptance probability of $\A$ over $u$
as $P_\A^u(\Last(F))$.

{A classic decision problem for probabilistic word automata is the \concept{value $1$ problem}.}

\begin{center}\fbox{
\begin{tabular}[t]{rl}
\textbf{INPUT:} & A probabilistic word automaton $\A$ and a subset $F \subseteq Q$ \\
\textbf{QUESTION:} & $\forall \varepsilon > 0, \exists u \in \Sigma^*,\ P_\A^u(\Last(F)) \ge 1 - \varepsilon$?
\end{tabular}
}\end{center}

\smallskip

Informally, {the value $1$ problem} asks for the existence of words {with  acceptance probabilities that are arbitrarily close to $1$. In this case,} we say that $\A$ has value $1$.
The undecidability of the value~$1$ problem for simple probabilistic automata was first established in~\cite{GO10} (see also~\cite{FGKO15} and~\cite{Fijalkow17} for a simple proof).

\begin{theorem}[\cite{GO10}]\label{thm:value1}
The value $1$ problem for simple probabilistic word automata is undecidable.
\end{theorem}

\subsection{Probabilistic Automata on Infinite Words}\label{sec:probabilistic-automata-infinite-words}
%-------------------------------------------------------------------------------------------
Baier, Größer, and Bertrand conducted an in-depth study of probabilistic automata over infinite words~\cite{BGB12}.
To define the semantics of a probabilistic word automaton $\A=(Q, q_\init, \delta)$ over an infinite word $w=w_1w_2\cdots$, we proceed as for finite words  
and construct a Markov chain $\M_\A^w$ whose set of states is $Q \times \N$. The initial state is again $(q_\init,0)$, and the probability transition function $T_\A^w$ is still defined by 
$$T_\A^w((p,i)) = \sum_{q \in Q} \delta(p,w_i)(q) \cdot (q,i+1)$$

A \concept{run} of $\A$ on $w$ is now an infinite path in $\M_\A^w$ and the Markov chain yields a probability measure $P_\A^w$ on runs.

{For probabilistic automata on infinite words} we mostly focus on the co-Büchi acceptance condition that
is defined as follows. Given a subset of states $F \subseteq Q$, we
let
$\CoBuchi(F)=(Q\times \mathbb{N})^*((Q\setminus F)\times
\mathbb{N})^\omega$ be the (measurable) set of runs that visit $F$
only finitely often, and, when this set of runs has measure $1$, we
say that $w$ is almost-surely accepted by $\mathcal{A}$ for the
co-Büchi condition $F$, {written $w \in L^{= 1}_{\CoBuchi(F)}(\A)$. Formally,}
$$
L^{= 1}_{\CoBuchi(F)}(\A) = \set{w \in \Sigma^\omega : P_\A^w(\CoBuchi(F)) = 1}.
$$

\begin{example}\label{ex:infinitelymanysharps}
	Let $\Sigma$ be an alphabet and $\sharp \notin \Sigma$ be a fresh symbol.
Let $\C$ be the simple probabilistic co-Büchi automaton {with set {$\set{p_1,p_2}$ of states},  
  initial state $p_1$, and transition function given by:}
  \begin{itemize}
  	\item[-] $\delta(p_1,a) = p_1$ for any {$a\in\Sigma \setminus \set{\sharp}$};
  	\item[-] $\delta(p_1,\sharp) = \frac{1}{2} p_1 + \frac{1}{2} p_2$; and
  	\item[-] $\delta(p_2,a) = p_2$ for any $a\in\Sigma \cup \set{\sharp}$.
  \end{itemize}
  As $p_2$ is absorbing and as moving from $p_1$ to $p_2$ may only
  happen when reading $\sharp$, {the language} 
  $L^{= 1}_{\CoBuchi(\{p_1\})}(\C)$ consists of those infinite words
  over $\Sigma \cup \set{\sharp}$ that contain infinitely many
  occurrences of $\sharp$.
  {Note that we will later use this example as a gadget in the proof of Proposition~\ref{prop-simple-cobuchi}}
\end{example}

The \concept{emptiness problem} for probabilistic co-Büchi word automata with almost-sure semantics is the following decision problem:

\smallskip
\begin{center}\fbox{
\begin{tabular}[t]{rl}
\textbf{INPUT:} & A probabilistic word automaton $\A$ and a set $F \subseteq Q$ \\
\textbf{QUESTION:} & Is $L^{= 1}_{\CoBuchi(F)}(\A) = \emptyset$?
\end{tabular}
}\end{center}
\smallskip

It was shown in~\cite{BGB12} that this problem is
undecidable.

\begin{proposition}[\cite{BGB12}]
\label{prop-cobuchi}
The emptiness problem for probabilistic co-Büchi word automata with almost-sure semantics is undecidable.
\end{proposition}

The proof in~\cite{BGB12} is obtained by reducing the universality problem for simple probabilistic \emph{Büchi} word automata with the \emph{positive semantics}: Indeed, automata in this class (we refer to~\cite{BGB12} for definitions) can be effectively complemented into probabilistic co-Büchi word automata with the almost-sure semantics, and whose universality problem is proved to be undecidable. As the complementation procedure does not preserve the property of being simple, we will {later argue (see Proposition~\ref{prop-simple-cobuchi})} that Proposition~\ref{prop-cobuchi} still holds for simple probabilistic co-Büchi word automata with almost-sure semantics.

\subsection{Universal Automata on {Infinite Trees}  with Qualitative Semantics}
%---------------------------------------------------------------------------------------
\label{sec:univ-tree-automata}

The qualitative semantics for tree automata was introduced by Carayol, Haddad, and Serre in~\cite{CHS14} and was studied 
for non-deterministic~\cite{CHS14}, alternating~\cite{fijalkow2013emptiness}, and probabilistic automata~\cite{CHS14}.

{In this section, we define universal tree automata with qualitative semantics and then extend this concept to alternating tree automata with qualitative semantics in the next section.}

A \concept{tree automaton} is a tuple $\A = (Q, q_\init, \Delta)$, where $Q$ is a finite set of states,
$q_\init$ is the initial state, and $\Delta \subseteq Q \times \Sigma \times Q \times Q$ is the transition relation. 
A \concept{run} of $\A$ over a $\Sigma$-tree $t$ is a $Q$-tree $\rho : \set{0,1}^* \to Q$
such that $\rho(\varepsilon) = q_\init$ and, for all $u \in \set{0,1}^*$, we have $(\rho(u),t(u),\rho(u 0),\rho(u 1)) \in \Delta$.
We let $\Run_\A(t)$ denote the set of runs of $\A$ over $t$.

A tree automaton $\A$ and a run $\rho$ induce a Markov chain $\M_{\A}^\rho$ as follows.
The set of states is $Q \times \set{0,1}^*$, the initial state is $(q_\init,\varepsilon)$, and the probability transition function $T_\A^\rho$ is given by 
$$T_\A^\rho((\rho(u),u)) = \frac{1}{2}  (\rho(u 0),u 0) + \frac{1}{2}  (\rho(u 1),u 1)$$
yielding the probability measure $P_\A^\rho$ {on branches of
the run $\rho$}.

Given a subset of states $F \subseteq Q$, we let $\CoBuchi(F)=(Q\times\{0,1\}^*)^*((Q\setminus F)\times\{0,1\}^*)^\omega$ be the (measurable) set of infinite paths in $\M_{\A}^\rho$ that visit $F$ only finitely often, and we say that the run $\rho$ is \concept{qualitatively accepting} for the co-Büchi condition $F$ if $P_\A^\rho(\CoBuchi(F)) = 1$.
Equivalently, a run $\rho$ is qualitatively accepting for the co-Büchi condition if and only if the set of branches in $\rho$ that contain finitely many nodes labelled by a state in $F$ has measure $1$ for the classical \emph{coin-flipping} measure $\mu$ on branches: $\mu$ is the unique complete probability measure such that $\mu(u\cdot\{0,1\}^\omega)=2^{-|u|}$.

The universal semantics yields the following definition:
\[
L^{\forall}_{\text{Qual},\CoBuchi(F)}(\A) = \set{t \in \Tree(\Sigma) : \forall \rho \in \Run_\A(t), P_\A^\rho(\CoBuchi(F)) = 1}.
\]

In words, a tree $t$ belongs to $L^{\forall}_{\text{Qual},\CoBuchi(F)}(\A)$ if every run of $\A$ over $t$ 
is such that almost all its branches contain finitely many states in $F$. 

The \concept{emptiness problem} for universal co-Büchi tree automata with qualitative semantics is the following decision problem:

\smallskip
\begin{center}\fbox{
\begin{tabular}[t]{rl}
\textbf{INPUT:} & A tree automaton $\A$ and a set $F \subseteq Q$\\
\textbf{QUESTION:} & Is $L^{\forall}_{\text{Qual},\CoBuchi(F)}(\A) = \emptyset$?
\end{tabular}
}\end{center}
\smallskip

We will prove in Theorem~\ref{thm:undecidability-univcoB} that this problem is undecidable.

\subsection{Alternating Automata on Infinite Trees with Qualitative Semantics}
%----------------------------------------------------
\label{sec:alt-tree-automata}

An \concept{alternating tree automaton} is a tuple $\A = (Q, q_\init,Q_E,Q_A,\Delta)$, where $Q$ is the finite set of states, $q_\init$ is the initial state, $(Q_E,Q_A)$ is a partition of $Q$ into \Eloise's and \Abelard's states and $\Delta \subseteq Q \times \Sigma \times Q \times Q$ is the transition relation. 

The input of such an automaton is a $\Sigma$-tree $\tree$ and
acceptance is defined {by means of} the following two-player perfect-information stochastic game $\game_{\A,t}^{=1}$.
Intuitively, a play in this game consists in moving a pebble along a branch of
$t$ {starting from the root}: the pebble is attached to a state and in a
node $u$ with state $q$, \Eloise (if $q \in Q_E$) or \Abelard (if
$q \in Q_A$) picks a transition $(q,t(u),q_0,q_1) \in \Delta$, and then \Random chooses to move down the pebble either to {node} $u0$ 
(and {then updates} the state to $q_0$) or to {node} $u1$ 
(and {then updates} the state to $q_1$). 

Formally, let $G = (V_E \cup V_A \cup V_R,E)$ 
with $V_E = Q_E\times\{0,1\}^*$, $V_A =Q_A\times \{0,1\}^*$ 
and $V_R = \{(q,u,q_0,q_1) \mid u \in \{0,1\}^* \text{ and } (q,t(u),q_0,q_1) \in \Delta\}$, and
$$\begin{array}{ll}
E \qquad = \qquad & \{((q,u),(q,u,q_0,q_1)) \mid (q,u,q_0,q_1) \in V_R\} \quad \cup \\ 
&\{((q,u,q_0,q_1),(q_x,u \cdot x)) \mid x \in \{0,1\} \text{ and } (q,u,q_0,q_1) \in V_R\}\ 
\end{array}$$

Then, we define $\arena_{\A,t}^{=1} = (G,\VE,\VA,\VR,\delta,(q_\init,\epsilon))$ where $\delta((q,u,q_0,q_1))=\frac{1}{2}(q_0,u0) + \frac{1}{2}(q_1,u1)$.

Given a subset of states $F\subseteq Q$, we say that $t$ is qualitatively accepted by $\A$ for the Büchi (\emph{resp.} co-Büchi) condition $F$ if \Eloise has an almost-surely winning strategy in the Büchi (\emph{resp.} co-Büchi) game $\game_{\A,t}^{=1}=(\arena_{\A,t}^{=1},F\times\{0,1\}^*)$. 

For an alternating tree automaton $\A$ and a subset of states $F$, we denote by $L^{\text{Alt}}_{\text{Qual},\Buchi(F)}(\A)$ (\emph{resp.} $L^{\text{Alt}}_{\text{Qual},\CoBuchi(F)}(\A)$) the set of trees qualitatively accepted by $\A$ for the Büchi (\emph{resp.} co-Büchi) condition $F$. 

\begin{remark}\label{rk:positionalStrategies}
Any \emph{positional} strategy for \Eloise in $\game_{\A,t}^{=1}$ 
can be described as a function $\sigma: Q_E\times \{0,1\}^* \rightarrow  Q \times Q$ 
that satisfies the following property: $\forall u \in \{0,1\}^*$, 
if $\sigma(q,u) = (q_0,q_1)$ then $(q,t(u),q_0,q_1) \in \Delta$. 
Equivalently, in a {curried} form, $\sigma$ is a map $\{0,1\}^* \rightarrow (Q_E \rightarrow Q \times Q)$. 
Hence, if one lets $\mathcal{T}$ be the set of functions from $Q_E$ into $Q \times Q$, 
\Eloise's positional strategies are in bijection with $\mathcal{T}$-labelled binary trees.
\end{remark}

It is easily seen that universal tree automata with qualitative
semantics are subsumed by alternating tree automata with qualitative
semantics. Indeed we have the following classical result (that we
state here only for {co-Büchi} acceptance condition but that
works similarly for any other acceptance condition).

\begin{proposition}\label{proposition:altvsuniv}
	Let $\A = (Q, q_\init, \Delta)$ be a tree automaton and let $F\subseteq Q$. Consider the alternating tree automaton $\B = (Q, q_\init,\emptyset,Q,\Delta)$, meaning that all states of $\A$ are interpreted as \Abelard's. Then the following holds.
	$$L^{\forall}_{\text{Qual},\CoBuchi(F)}(\A) = L^{\text{Alt}}_{\text{Qual},\CoBuchi(F)}(\B) $$%.
\end{proposition}

\begin{proof} 
  For a fixed tree $t$, runs of $\A$ over $t$ are in bijection with
  strategies of \Abelard in the co-Büchi game $\game_{\B,t}^{=1}$
  (where \Eloise is making no choice), and moreover a run is
  qualitatively accepting for $\A$ if and only if \Eloise
  almost-surely wins in $\game_{\B,t}^{=1}$ when \Abelard uses the
  corresponding strategy. Hence, all runs of $\A$ over $t$ are
  qualitatively accepting if and only if \Eloise almost-surely wins
  against {every} strategy of \Abelard in $\game_{\B,t}^{=1}$, which
 means that
  $L^{\forall}_{\text{Qual},\CoBuchi(F)}(\A) =
  L^{\text{Alt}}_{\text{Qual},\CoBuchi(F)}(\B)$.
\end{proof}

The \concept{emptiness problem} for alternating Büchi  tree automata with qualitative semantics is the following decision problem:

\smallskip
\begin{center}\fbox{
\begin{tabular}[t]{rl}
\textbf{INPUT:} & An {alternating} tree automaton $\A$ and a set $F \subseteq Q$\\
\textbf{QUESTION:} & Is $L^{\text{Alt}}_{\text{Qual},\Buchi(F)}(\A) = \emptyset$?
\end{tabular}
}\end{center}
\smallskip

We will prove in Theorem~\ref{thm:decidability-altBuchi} that this problem is decidable in exponential time.

\begin{remark}
	The emptiness problem can be {similarly} defined for alternating co-Büchi tree automata with qualitative semantics. However, {this problem is undecidable as a corollary of  Proposition~\ref{proposition:altvsuniv} together with the forthcoming Theorem~\ref{thm:undecidability-univcoB},  proving the undecidability of the emptiness problem for universal co-Büchi tree automata with qualitative semantics.}
\end{remark}

\subsection{Probabilistic Automata on Infinite Trees with Qualitative Semantics}
%------------------------------------------------------
\label{def:prob_tree_automata}

Probabilistic tree automata with qualitative semantics were defined
in~\cite{CHS14} {with the intention of lifting the
definition of probabilistic automata on infinite words to the case of infinite
trees}. In particular, an input tree induces a probability distribution
over runs and acceptance is defined by requiring that almost all runs
should be accepting. {Mixed} with the qualitative {co-Büchi} semantics, this means
that a tree is accepted if almost all runs have almost all their
branches containing finitely many states from $F$. {Contrary to the
authors of \cite{CHS14} who define a probability measure on runs, we follow another approach (still yielding an equivalent notion
\cite[Proposition~45]{CHS14}) based on Markov chains.}

A \concept{probabilistic tree automaton} is a tuple $\A = (Q, q_\init, \delta)$, where $Q$ is the finite set of states,
$q_\init$ is the initial state, and $\delta : Q \times \Sigma \to \D(Q \times Q)$ is the transition function.

A probabilistic tree automaton $\A$ and a tree $t$ induce a Markov chain $\M_{\A}^t$ as follows.
The set of states is $Q \times \set{0,1}^*$, the initial state is $(q_\init,\varepsilon)$, and the probability transition function $T_\A^t$ is given by (where $\cdot$ distributes over $+$)
{
\[
T_\A^t((q,u)) = \sum_{q_0,q_1 \in Q} \delta(q,t(u))(q_0,q_1) \cdot \left(
  \frac{1}{2} {\cdot} (q_0, u 0) + \frac{1}{2} {\cdot} (q_1, u 1) \right),
\]
}

Given a subset of states $F \subseteq Q$, we again let $\CoBuchi(F)=(Q\times\{0,1\}^*)^*((Q\setminus F)\times\{0,1\}^*)^\omega$ be the (measurable) set of infinite paths in $\M_{\A}^t$ that visit $F$ only finitely often. Then the probability measure $P_\A^t$ induced by $\M_{\A}^t$ yields the following definition of the set of trees \concept{almost-surely qualitatively accepted} by $\A$: 
\[
L^{\forall^{=1}}_{\text{Qual},\CoBuchi(F)}(\A) = \set{t \in \Tree(\Sigma) : P_\A^t(\CoBuchi(F)) = 1}.
\]

{We now turn to our main decidability result about emptiness of alternating B{\"u}chi tree automata with qualitative semantics.}

\section{Decidability of the Emptiness Problem for Alternating Büchi Tree Automata with Qualitative Semantics}
\label{sec-dec-Buchi}
In this section, we prove Theorem~\ref{thm:decidability-altBuchi} that states the  decidability of the emptiness problem for alternating 
Büchi tree automata with qualitative semantics, which contrasts with the forthcoming result that the emptiness problem for universal co-Büchi tree automata with qualitative semantics is undecidable (Theorem~\ref{thm:undecidability-univcoB} of  Section~\ref{sec-undec-automata}).

Our approach for checking emptiness of an alternating Büchi tree automaton
$\A$ with qualitative semantics relies on a two-player
\emph{imperfect-information} stochastic \emph{finite} Büchi game. In this game, \Eloise almost-surely wins if, and only if, the language accepted
by $\A$ is non-empty. As for this class of games, one can decide 
whether \Eloise has an almost-surely winning strategy, the
announced decidability result follows. 

We establish in
Section~\ref{section:positionality} a preliminary general positionality result to be used in Section~\ref{sec:emptiness-checking} for proving the
equivalence between {\Eloise} almost-surely winning in the game and 
 {$\A$ accepting} some tree.
 
\subsection{A Positionality Result for Chronological Games}
%--------------------------------------------------------
\label{section:positionality}

For the rest of this section, we fix a {stochastic} arena
$\arena=(G,V_E,V_A,V_R,\delta,v_\init)$ with $G=(V,E)$. Moreover, we assume that the game is \emph{chronological} in the sense that there exists a function
$\rank: V_E\cup V_A\cup V_R \rightarrow \N$ such that
$\rank^{-1}(0)=\{v_\init\}$ and for $(v,v') \in E$,
$\rank(v') = \rank(v) + 1$.  Note that the
arena $\arena_{\A,t}^{=1}$ used in
Section~\ref{sec:alt-tree-automata} to define acceptance of a tree $t$ by
an alternating tree automaton with qualitative semantics $\A$ is chronological. Note also that a chronological arena with finite out-degree has a countable set of vertices.

\begin{theorem}\label{theo:positionalityChrono}
In a two-player perfect-information stochastic Büchi game played on a chronological arena with finite out-degree, \Eloise has an almost-surely winning strategy if, and only if, she has a positional almost-surely winning strategy.
\end{theorem}

Actually, the {core difficulty lies in proving Theorem~\ref{theo:positionalityChrono} for the simple case of reachability games.}

\begin{theorem}\label{theo:positionalityChronoReach}
{In a two-player perfect-information stochastic reachability game played on a chronological arena with finite out-degree, \Eloise has an almost-surely winning strategy if, and only if, she has a positional almost-surely winning strategy.}
\end{theorem}

\begin{proof}
{The direction from right to left is immediate}. For the other direction, the key steps are the following. First, we establish (Lemma~\ref{lemma:seuil}) that if \Eloise can ensure to reach
  $F$ with probability~$1$ from some initial vertex, then there exists
  a bound $k$ such that she can ensure to reach $F$ with probability
  at least half within $k$ steps.  {Second, we exploit Lemma~\ref{lemma:seuil}} to “slice” the arena into infinitely many disjoint
  finite arenas: in each slice \Eloise plays to reach $F$ with
  probability at least half. Since each slice forms a finite
  sub-arena, optimal \emph{positional} strategies always
  exist. Finally, the strategy that plays in turns the
  latter positional strategies ensures to almost-surely reach $F$ in the long run.

{Let $\game=(\arena,F)$ be a two-player perfect-information stochastic reachability game played on a chronological arena with finite out-degree. In the following, a strategy in $\game$ from a vertex $v$ is a strategy in the game obtained from $\game$ by changing the initial vertex of the arena $\arena$ to $v$.}

The following lemma allows us to decompose the infinite arena $\arena$ into infinitely many finite arenas.

\begin{lemma}\label{lemma:seuil}
Let $\sigma_\Ei$ be an almost-surely winning strategy for \Eloise in
$\game$ from some vertex $v$. 
Then, there exists an integer $k$ such that for any  strategy $\sigma_\Ai$ of \Abelard,
we have
$$\proba_{\sigma_\Ei,\sigma_\Ai}(V^{\leq k}FV^\omega)\geq \frac{1}{2}\ .$$
\end{lemma}

\begin{proof}[Proof of Lemma~\ref{lemma:seuil}]
  Toward a contradiction, assume that such a $k$ does not exist.
  Hence, for each $k$ there exists a strategy $\sigma_{\Ai,k}$ such
  that
  $\proba_{\sigma_\Ei,\sigma_{\Ai,k}}(V^{\leq k}FV^\omega)<
  \frac{1}{2}$.  
  
  {Without loss of generality, we can assume that 
  $\sigma_{\Ai,k}$ is positional. Indeed, one can pick for
  $\sigma_{\Ai,k}$ a strategy for \Abelard that minimises the probability of winning for \Eloise in the
  reachabililty game obtained by restricting $\game$ to vertices of rank at most $k$. This game has a finite arena since $\game$ has finite out-degree, and by e.g.\ \cite{GZ07} such a strategy for \Abelard can be chosen positional.}

  From the sequence of strategies 
  $(\sigma_{\Ai,k})_{k \geq 0}$, we {now extract} a strategy
  $\sigma_{\Ai,\infty}$ {(designed to contradict the assumption that \Eloise has an almost-surely winning strategy)} that for every $k \geq 0$, {agrees} with
  infinitely many $\sigma_{\Ai,h}$ on its first $k$ moves.  {Since $\game$ has countably many vertices, fix} an (arbitrary) enumeration $v_1,v_2,\cdots$ of the vertices in
  $V$.

  {We define $\sigma_{\Ai,\infty}$ step-wise inductively on $i$: at step $i$, $\sigma_{\Ai,\infty}$ is defined on   $v_1,\cdots,v_i$ and on these vertices agrees with all those strategies $\sigma_{\Ai,h}$ with $h \in I_i$ where the sequence
  $I_0 \supseteq I_1 \supseteq I_2 \supseteq I_3 \supseteq \cdots$ is
  also defined inductively on $i$ and is such that each $I_i$ is
  infinite.}
  
  We let $I_0 = \N$ be the set of all positive integers.  
  
  {For $I_i$ where ${i\geq 1}$, consider the values of
  $\sigma_{{\Ai},h}(v_i)$ for all $h \in I_{i-1}$. Because  $G$
  has finite out-degree, there is some $v$ such that  $\sigma_{\Ai,h}(v_i) = v$, for infinitely many
  $h \in I_{i-1}$.}  We define
  $\sigma_{\Ai,\infty}(v_i) = v$ and we let
  $I_i = \set{h \in I_{i-1} \mid \sigma_{\Ai,h}(v_i) = v}${; note that $I_i$ is infinite.}

{Now, for $k \geq 0$, it is easy to see that by choosing $i$ big enough so that all vertices of rank at most $k$ belong to $\set{v_1,\ldots,v_i}$, strategy  $\sigma_{\Ai,\infty}$ agrees on its $k$ first moves with the infinitely many $\sigma_{\Ai,h}$ where $h \in I_i$.}

{As a consequence, for every $k$} there is some $h\geq k$ such that 
$$\proba_{\sigma_\Ei,\sigma_{\Ai,\infty}}(V^{\leq k}FV^\omega) = 
\proba_{\sigma_\Ei,\sigma_{\Ai,h}}(V^{\leq k}FV^\omega) \leq 
\proba_{\sigma_\Ei,\sigma_{\Ai,h}}(V^{\leq h}FV^\omega)<\frac{1}{2}$$ 
As $V^*FV^\omega = \bigcup_{k\geq 0}V^{\leq k}FV^\omega$ and as the sequence $(V^{\leq k}FV^\omega)_{k\geq 0}$ 
is increasing for set inclusion, one concludes that 
$$\proba_{\sigma_\Ei,\sigma_{\Ai,\infty}}(V^*FV^\omega) = 
\lim_{k\rightarrow \infty} \proba_{\sigma_\Ei,\sigma_{\Ai,\infty}}(V^{\leq k}FV^\omega) \leq \frac{1}{2}<1$$
which leads to a contradiction with $\sigma_\Ei$ being almost-surely winning, {and concludes the proof of Lemma~\ref{lemma:seuil}}.
\end{proof}

{Keeping on with the proof of Theorem~\ref{theo:positionalityChronoReach},} assume that \Eloise has an almost-surely wining strategy $\sigma_\Ei$ in $\game$.
Without loss of generality, we can assume that she has an
almost-surely {winning} strategy from everywhere,
by restricting the arena to vertices reachable by an almost-surely winning strategy.

For $k < k'$, we define the reachability game $\game_{[k,k']}$  induced by restricting the {arena $\arena= (G,\VE,\VA,\VR,\delta,\vini)$} to vertices of rank in $[k,k']${ where we add self-loops on vertices of rank $k'$ to avoid having dead-end vertices}.
Since $G$ has finite out-degree, there are finitely many vertices of rank in $[k,k']$, hence  $\game_{[k,k']}$ is finite.

We define inductively an increasing sequence of ranks $(k_i)_{i \ge 1}$ together 
with a sequence of strategies $(\sigma_{\Ei,[k_i,k_{i+1}[})_{i \ge 1}$
 such that for all $i \ge 1$, 
$\sigma_{\Ei,[k_i,k_{i+1}[}$ is a positional strategy, defined on all
vertices of rank in $[k_i,k_{i+1}[$, and such that 
from all vertices of rank $k_i$, for all strategies $\sigma_\Ai$, we have
$$\proba_{\sigma_{\Ei,[k_i,k_{i+1}[},\sigma_\Ai}(V^{\leq \ell}FV^\omega)\geq \frac{1}{2}\ ,$$	
where $\ell = k_{i+1} - k_i$.

Assume the first $i$ ranks and strategies are defined.
For each vertex of rank $k_i$, Lemma~\ref{lemma:seuil} gives
the existence of {some bound $k$}; since there are finitely many such vertices,
we can consider the maximum of those bounds that we call
{$\ell$, and we let $k_{i+1}=k_i+\ell$}.
By construction and {Lemma~\ref{lemma:seuil}}, from all vertices of rank $k_i$, for all strategies $\sigma_\Ai$, we have
$$\proba_{\sigma_\Ei,\sigma_\Ai}(V^{\leq \ell}FV^\omega)\geq \frac{1}{2}\ ,$$	
where $\ell = k_{i+1} - k_i$.
In other words, \Eloise wins the reachability game $\game_{[k_i,k_{i+1}]}$
with probability at least half,
so, relying on a generalisation\footnote{More precisely, when playing a reachability game on a finite arena, \Eloise always has an optimal positional strategy, where $\sigma_\Ei$ being optimal means that $\inf_{\sigma_\Ai} \proba_{\sigma_\Ei,\sigma_\Ai}(V^{\leq \ell}FV^\omega) = \sup_{\sigma'_\Ei}\inf_{\sigma_\Ai} \proba_{\sigma_\Ei,\sigma'_\Ai}(V^{\leq \ell}FV^\omega)$.} of Theorem~\ref{theo:posDetSt} (see
\emph{e.g.} \cite{GZ07,Kucera_foxes}), 
there exists an optimal uniform {(i.e. working from any initial vertex)}  positional strategy,
that we call $\sigma_{\Ei,[k_i,k_{i+1}[}$.
This concludes the inductive construction.

Now, define $\sigma_{\Ei,\infty}$ as the disjoint union of the
strategies $\sigma_{\Ei,[k_i,k_{i+1}[}$.  This is a positional
strategy; we argue that it is almost-surely winning.  Assume, towards
a contradiction, that this is not the case. Then, there exists
$\varepsilon > 0$ and a strategy $\sigma_{\Ai}$ such that
$$\proba_{\sigma_{\Ei,\infty},\sigma_\Ai}(V^{*}FV^\omega) \le 1 - \varepsilon\ .$$	
Observe that playing consistently with the first $p$ strategies $\sigma_{\Ei,[k_i,k_{i+1}[}$
ensures to reach $F$ with probability at least $1 - \frac{1}{2^p}$.
Since playing consistently with $\sigma_{\Ei,\infty}$ implies
playing consistently with the first $p$ strategies $\sigma_{\Ei,[k_i,k_{i+1}[}$,
we reach a contradiction by considering $p$ large enough so that $\frac{1}{2^p} < \varepsilon$.
\end{proof}

Theorem~\ref{theo:positionalityChrono} is an easy consequence of Theorem~\ref{theo:positionalityChronoReach} thanks to a simple and neat reduction from~\cite[Remark~2.3]{DBLP:journals/tocl/Chatterjee014} (also see \cite[Lemma~8.3]{BGB12}).
Roughly speaking, to turn a Büchi game into a reachability game equivalent with respect to almost-sure winning,
one adds a unique final vertex and replaces every Büchi vertex by a fresh random vertex which either reaches the final vertex
or proceeds in the game, each with probability half.
Then, visiting infinitely many Büchi vertices ensures to almost-surely reach the final vertex,
and conversely, reaching almost-surely the final vertex requires to almost-surely visit infinitely many Büchi vertices.

We make all this more formal.

\begin{proof}[Proof of Theorem~\ref{theo:positionalityChrono}]
  Recall that we denote by $\arena=(G,\VE,\VA,\VR,\delta,\vini)$,  with
  $G=(V,E)$, the underlying arena of $\game$ and denote by
  $F\subseteq V$ the set of vertices defining the Büchi condition. We
  now build an arena $\arena'=(G',\VE',\VA',\VR',\delta',\vini')$,  with
  $G'=(V',E')$, and a set $F'\subseteq V'$ of vertices such that
  \Eloise almost-surely wins  in the Büchi
  game $\game=(\arena,F)$ if and only if she almost-surely wins  in the reachability game
  $\game'=(\arena',F')$, {and in addition, if she has a
    positional almost-surely winning strategy in one game, she has one in the other}. This permits to deduce
  Theorem~\ref{theo:positionalityChrono} from
  Theorem~\ref{theo:positionalityChronoReach}.

{We formally explain how to construct $\arena$', taking care that it is chronological.}
    {The set of vertices $V'$ consists of $V$ augmented with a countable set of vertices $\{f_i\mid i\geq 0\}$, and with extra random vertices $F_R=\{v_s\mid s\in F\}$, one per
  {vertex} in $F$. The vertex $f_0$ has a unique outgoing transition to $f_{1}$ and it can be reached only from
  vertices in $F_R$. For every $i\geq 1$, the vertex $f_i$ has a unique outgoing transition to $f_{i+1}$ and it can be reached only from
   $f_{i-1}$. From a vertex $v_s\in F_R$ there are two outgoing
  edges: one to $f_0$ and one to $s$ and both can be chosen with the
  same probability half,
  i.e. $\delta'(v_s)=\frac{1}{2}f_0+\frac{1}{2}s$. Any edge in
  $G$ going from a vertex ${v\in V}$ to a vertex $s\in F$ is replaced by an
  edge from $v$ to $v_s$,  and if $v\in \VR$ we let $\delta'(v)(s)=0$ and
    $\delta'(v)(v_s)=\delta(v)(s)$. All other edges are left
    untouched: for every
  $v\in \VR$ and $s\notin {F \cup F_R \cup \{f_i\mid i\geq 0\}}$, we let
  $\delta'(v)(s)=\delta(v)(s)$.   Finally we let $\VE'=\VE\cup\{f_i\mid i\geq 0\}$, $\VA'=\VA$,
  $\VR'=\VR\cup F_R$ and $F'=\{f_0\}$.}
{Note that $\arena'$ is chronological by construction and $\arena$ being chronological.}

{There is an obvious correspondence between strategies (of
    both \Eloise and \Abelard) in $\game$ and strategies in $\game'$,
    and it preserves positionality.  Moreover, \Eloise almost-surely
    reaches the final state $f_0$ in $\game'$ {with strategy
      $\sigma'_\Ei$} if and only if she almost-surely visits
    infinitely often $F$ in $\game$ with {the corresponding} strategy
    $\sigma_\Ei$. Indeed, if she almost-surely visits $F$ in $\game$
    using $\sigma_\Ei$, due to positive transition probability to $f_0$
    from states in $F$, she almost-surely reaches $f_0$ in $\game'$
    using {$\sigma'_\Ei$}. Conversely, if against any
    strategy $\sigma_\Ei$ of \Eloise in $\game$, \Abelard has a strategy
    $\sigma_\Ai$ that ensures that $F$ is visited finitely often with
    some positive probability $\epsilon>0$, then  in $\game'$, when
    \Eloise and \Abelard use the corresponding pair of
    strategies $(\sigma'_\Ei,\sigma'_\Ai)$,  there is a
    positive probability $\epsilon'$  that  $f_0$ is never reached,  as the only way of reaching $f_0$
    is by going through $F$; hence, in $\game'$, against any 
    strategy $\sigma'_\Ei$ of \Eloise, \Abelard has a strategy $\sigma'_\Ai$ that
    avoids reaching $f_0$ with positive probability.}
\end{proof}

\begin{remark}\label{rk:purevsrandomised}
{As already announced,} in this paper we only considered pure (i.e. non-randomised) strategies. Hence, “\Eloise has an almost-surely winning strategy” should be understood in both Theorem~\ref{theo:positionalityChronoReach} and Theorem~\ref{theo:positionalityChrono} as “\Eloise has an almost-surely winning \emph{pure} strategy”. However, our proof directly carries over to the more general case of randomised strategies.
\end{remark}

\subsection{Checking Emptiness}\label{sec:emptiness-checking}

Fix an alternating tree automaton $\A = (Q, q_\init,Q_E,Q_A,\Delta)$
and a subset $F\subseteq Q$ of final states. 
In order to check whether $L^{\text{Alt}}_{\text{Qual},\Buchi(F)}(\A) = \emptyset$, we design an
imperfect-information stochastic Büchi game  $\game_\A^\emptyset$ in which \Eloise has an almost-surely winning strategy if and only if $L^{\text{Alt}}_{\text{Qual},\Buchi(F)}(\A) \neq \emptyset$. {The equivalence is proved by applying the positionality result established in Theorem~\ref{theo:positionalityChrono} to the acceptance game for $\mathcal{A}$.}

In the game, \Eloise describes both a tree $t$
and a positional strategy {$\sigma_t$} for her in the game
{$\game_{\A,t}^{=1}$}. Following Remark
\ref{rk:positionalStrategies}, the {positional} strategy {$\sigma_t$} is described as a
$\mathcal{T}$-labelled tree, where $\mathcal{T}$ denotes the set of
functions from $Q_E$ into $Q\times Q$. As the plays are of $\omega$-length,
\Eloise actually does not fully describe $t$ and $\sigma_t$ but only a branch:
this branch is chosen by \Random while \Abelard  takes care of computing
the sequence of states along it (either by updating an existential
state according to {$\sigma_t$} or, when the state is universal, by
choosing an arbitrary valid transition of the automaton). In this game, 
\Eloise observes the directions, but not the actual control state of the automaton. 

\begin{remark}The fact that \Eloise does not observe the control state of the automaton is crucial here, as it avoids her to cheat when describing the input tree. Indeed, consider an alternating tree automaton whose initial state belongs to \Abelard and from which there are two possible transitions: one that makes the automaton check that both subtrees only contain nodes labelled by $a$, and one that makes the automaton check  that  both subtrees only contain nodes labelled by $b$. Trivially, no tree is accepted by such an automaton. However, if one plays a modified version of the previous game where \Eloise observes the control state she can surely win in this game by producing a tree with all nodes labeled by $a$ (\emph{resp.} by $b$) depending on the initial choice by \Abelard.\end{remark}

Formally, we let $\arena_{\A}^{\emptyset}=(\states,A,T,\sim,\sini)$ where\begin{itemize}
	\item $\states = (Q\times \{0,1\})\cup \{(q_\init,\epsilon)\}$;
	\item $\sini = (q_\init,\epsilon)$;
	\item $A\subseteq \Sigma\times \mathcal{T}$ is the set 
{$\{(a,\tau) \mid \forall q\in Q_E,
\;(q,a,q_0,q_1)\in \Delta \text{ where }(q_0,q_1)=\tau(q)\}$};
	\item 
  $T=  \{((q,i),(a,\tau),d_{q_0,q_1})\mid q\in Q_E \text{ and }\tau(q)=(q_0,q_1)\} 
  \;\cup$

$ \{((q,i),(a,\tau),d_{q_0,q_1})\mid
  q\in Q_A \text{ and }(q,a,q_0,q_1)\in \Delta\}$     
where $d_{q_0,q_1}=\frac{1}{2}(q_0,0)+\frac{1}{2}(q_1,1)$; and
\item $(q,i)\sim(q',i)$ for all $q,q'\in Q$ and $i\in\{0,1\}$. 
\end{itemize}
Finally we let $\game_\A^{\emptyset} = (\arena_\A^\emptyset,F\times\{0,1\})$.  

The following theorem relates $\game_\A^{\emptyset}$ and $L^{\text{Alt}}_{\text{Qual},\Buchi(F)}(\A)$.

\begin{theorem}\label{theo:main-qual}
\Eloise almost-surely wins in $\game_\A^{\emptyset}$ iff $L^{\text{Alt}}_{\text{Qual},\Buchi(F)}(\A)\neq\emptyset$.
\end{theorem}

\begin{proof}
  Due to how $\sim$ is defined, a strategy for \Eloise in
  $\game_\mathcal{A}^{\emptyset}$ can also be viewed as a map
  $\sigma: \{0,1\}^*\rightarrow A$. As
  $A\subseteq \Sigma\times \mathcal{T}$, one can see $\sigma$ as a
  pair $(t,\sigma_t)$ where $t$ is an infinite $\Sigma$-labelled
  binary tree, and $\sigma_t$ is a \emph{positional} strategy for
  \Eloise in the acceptance game $\game_{\mathcal{A},t}^{=1}$. Now,
  once such a strategy $\sigma$ is fixed, the set of plays in
  $\game_\mathcal{A}^{\emptyset}$ where \Eloise respects $\sigma$ is
  in one-to-one correspondence with the set of plays in
  $\game_{\mathcal{A},t}^{=1}$ where she respects $\sigma_t$, and this
  correspondence preserves the property of being a winning
  play. Therefore, $\sigma=(t,\sigma_t)$ is almost-surely winning {in $\game_\A^{\emptyset}$} iff
  $\sigma_t$ is an almost-surely winning {positional}
  strategy in $\game_{\mathcal{A},t}^{=1}$ iff
  $t\in L^{\text{Alt}}_{\text{Qual},\Buchi(F)}(\A)$. {The
    last equivalence holds because}, thanks to
  Theorem~\ref{theo:positionalityChrono}, we can restrict our attention
  to positional strategies for \Eloise { in the perfect-information game} 
  $\game_{\mathcal{A},t}^{=1}$ {which, we recall, is chronological and of course has finite out-degree}. {Finally}, \Eloise has an almost-surely
  winning strategy {in $\game_\mathcal{A}^\emptyset$} iff there exists some tree
  $t\in L^{\text{Alt}}_{\text{Qual},\Buchi(F)}(\A)$.
\end{proof}

Combining Theorem~\ref{theo:main-qual} with Theorem~\ref{theo:AS_impInfo} directly implies decidability of the emptiness problem for alternating Büchi tree automata with qualitative semantics.

\begin{theorem}\label{thm:decidability-altBuchi}
The emptiness problem for alternating Büchi tree automata with qualitative semantics is decidable in exponential time.
\end{theorem}

Regarding lower bound, following the same ideas as in the undecidability proof in Theorem~\ref{thm:undecidability-univcoB}, one can reduce the emptiness problem for simple probabilistic \emph{Büchi} automata with almost-sure semantics to the emptiness problem for universal\footnote{Following Proposition~\ref{proposition:altvsuniv}, we call \emph{universal} an alternating Büchi tree automata whose set of states belonging to \Eloise is empty.} Büchi tree automata with qualitative semantics.

\begin{theorem}\label{thm:decidability-univBuchi-lowerbound}
The emptiness problem for universal Büchi tree automata with qualitative semantics is hard for ExpTime.
\end{theorem}

\begin{proof}
{Similarly to what was done in Section~\ref{sec:probabilistic-automata-infinite-words}
  for the co-Büchi acceptance condition, we define a probabilistic
  \emph{Büchi} automaton with almost-sure semantics on infinite words: for a probabilistic automaton $\A=(Q,q_\init,\delta)$ and a subset of states $F \subseteq Q$, we let
  $\Buchi(F)=\bigcap_{i\geq 0}Q^iQ^*FQ^\omega$ be the (measurable) set
  of runs that visit $F$ infinitely often. We then let:} 
$$
L^{= 1}_{\Buchi(F)}(\A) = \set{w \in \Sigma^\omega : P_\A^w(\Buchi(F)) = 1}.
$$

The \concept{emptiness problem} for probabilistic Büchi word automata with almost-sure semantics is the following decision problem:

\smallskip
\begin{center}\fbox{
\begin{tabular}[t]{rl}
\textbf{INPUT:} & A probabilistic word automaton $\A$ and a set $F \subseteq Q$ \\
\textbf{QUESTION:} & Is $L^{= 1}_{\Buchi(F)}(\A) = \emptyset$?
\end{tabular}
}\end{center}
\smallskip

{It is proved in \cite{BGB12} that this problem is complete for
ExpTime. Moreover, this result {still holds with } the extra
requirement that the automata are simple. Indeed, the lower bound in
\cite{BGB12} is by reduction of the almost-sure repeated reachability
for partial-observation Markov decision processes. This latter problem was shown to be ExpTime-complete by
de~Alfaro~\cite{DeAlfaro99}. The hardness proof in~\cite{DeAlfaro99}, based on the concept of  blindfold games as defined by Reif in his 
seminal paper~\cite{Reif84}, {survives (with the same proof) if the branching in the partial-observation Markov decision process has at most two states. Consequently, hardness for ExpTime already
holds for probabilistic automata whose distributions involved in the transition function 
have a support of at most two states}. Finally, as observed in
\cite[Remark~8.9]{BGB12}, emptiness is not affected by changing the
probabilities in the distributions {as long as} the support is unchanged: therefore, one can always reduce to the case of simple automata.}

Now, following exactly the same path as in Theorem~\ref{thm:undecidability-univcoB} one proves that the emptiness problem for simple probabilistic Büchi automata with almost-sure semantics can be polynomially reduced to the emptiness problem for universal Büchi tree automata with qualitative semantics, which implies the announced lower-bound.
\end{proof}

%{We now turn to our main undecidability result about the emptiness problem for universal co-Büchi tree automata.}

\section{Undecidability of the Emptiness Problem for Universal Co-Büchi Tree Automata with Qualitative Semantics}
\label{sec-undec-automata}

In this section we prove our main undecidability result on the emptiness problem for universal 
co-Büchi tree automata with qualitative semantics, 
from which we will then derive the undecidability of \MSOzeropath in Section~\ref{sec-cor}. 
We prove this result by reduction from the emptiness problem for
simple probabilistic co-Büchi word automata with almost-sure
semantics. As already mentioned (Proposition~\ref{prop-cobuchi}) it was shown in~\cite{BGB12} that this problem is
undecidable for general probabilistic word automata, but in our
reduction to probabilistic tree automata it will be crucial to work with simple
ones.
We thus start by giving a proof of this slightly stronger
result. 
 
\begin{proposition}
\label{prop-simple-cobuchi}
The emptiness problem for simple probabilistic co-Büchi word automata with almost-sure semantics is undecidable.
\end{proposition}

\begin{proof}
  The proof is by reduction from the value 1 problem for simple
  probabilistic automata, which is undecidable
  (Theorem~\ref{thm:value1}).

  Let $\A = (Q,q_\init,\delta)$ be a simple probabilistic word automaton over some alphabet $\Sigma$,
and let $F \subseteq Q$.
Let $\sharp\notin \Sigma$ be a fresh symbol and let $\A' = (Q\cup\{q'_\init\},q_\init,\delta')$ be the simple probabilistic automaton over $\Sigma\cup\{\sharp\}$ obtained from $\A$ as follows:
\begin{itemize}
\item[-] $q'_\init$ is a new state with  $\delta'(q_\init',a)=\delta(q_\init,a)$, for any letter $a\neq \sharp$, and $\delta'(q_\init,\sharp) = q_\init'$;
	\item[-] $\delta'(q,a) = \delta(q,a)$, for any state $q\in Q$ and any letter $a\neq \sharp$;
	\item[-] $\delta'(q,\sharp) = q_\init$ if $q\in F$ and $\delta'(q,\sharp) = q_\init'$ otherwise, for any state $q\in Q$.
\end{itemize}
{We equip $\A'$ with the co-Büchi condition $\{q_\init'\}$.} Note that $\A'$ is simple. 

For a sequence of words $(u_i)_{i\geq 1}$ over $\Sigma$ we let $x_i$ be
  the acceptance probability of $\A$ over $u_i$, for every
  $i\geq 1$.
Now consider an infinite word of the form
$w= \sharp u_1\sharp u_2\sharp u_3\cdots$, {and let  $E_i$ be
  the event: ``$\A'$ ends in $q_\init'$
when it reads $u_i\sharp$ from $q_\init$ or $q_\init'$''. Each $E_i$ has
probability $1-x_i$, and they are mutually independent. Also, $w$ is
almost-surely accepted by $\A'$ if and only if the probability that
infinitely many of the events $E_i$ occur is zero.} It is then  a
direct consequence of the Borel-Cantelli Lemma (and its converse) that
$w$ is almost-surely accepted by $\A'$  if and only if $\sum_{i=1}^\infty 1-x_i < \infty$. 

It follows that $\A$ has value $1$ if and only if $\A'$ almost-surely
accepts a word of the form $w= \sharp u_1\sharp u_2\sharp
u_3\cdots$. Indeed, if $\A$ has value $1$ then there is a sequence of
words $(u_i)_{i\geq 1}$ such that $x_i\geq 1-\frac{1}{i^2}$ and
therefore such that $\sum_{i=1}^\infty 1-x_i < \infty$; conversely, if
a sequence of words $(u_i)_{i\geq 1}$ is such that
$\sum_{i=1}^\infty 1-x_i < \infty$, one must have
$\lim_{x\to\infty} x_i = 1$.

To conclude the proof it is sufficient to build a simple probabilistic
co-Büchi word automaton $\B$ that almost-surely accepts only those
words that are almost-surely accepted by $\A'$, starting with a $\sharp$
and containing infinitely many $\sharp$.

Consider the automaton $\C$ from Example~\ref{ex:infinitelymanysharps} and recall that, when equipped with the acceptance condition $\CoBuchi(\{p_1\})$, it accepts those infinite words over $\Sigma\cup\{\sharp\}$ that contain infinitely many occurrences of $\sharp$.

Now, define $\B$ as the simple probabilistic automaton consisting of a fresh initial state $q_\init''$ together with a copy of $\A'$ and a copy of $\C$. From $q_\init''$ the only possible action is to read a $\sharp$ and go either to the initial state of $\B$ with probability $\frac{1}{2}$ or to the initial state of $\C$ with probability $\frac{1}{2}$. 

Then it is immediate that $L^{= 1}_{\CoBuchi(\{q_\init',p_1\})}(\B)$ is empty if and only if $\A$ does not have value~1.
\end{proof}

{Our main undecidability result of Theorem~\ref{thm:undecidability-univcoB} contrasts with two decidability
results, for probabilistic Büchi {tree}
automata~\cite{CHS14} and for alternating Büchi
{tree} automata~\cite{fijalkow2013emptiness} (Theorem~\ref{thm:decidability-altBuchi}), both with qualitative
semantics.}

\begin{theorem}\label{thm:undecidability-univcoB}
The emptiness problem for universal co-Büchi tree automata with qualitative semantics is undecidable.
\end{theorem}

To prove Theorem~\ref{thm:undecidability-univcoB} we construct a reduction 
from the emptiness problem for simple probabilistic co-Büchi word automata with almost-sure semantics to 
the emptiness problem for universal co-Büchi tree automata with qualitative semantics.
The correctness of the reduction relies on the two following results {(Lemma~\ref{lem:word_to_tree} and Lemma~\ref{lem:cross})}.

Let $\A = (Q,q_\init,\delta)$ be a simple probabilistic word automaton and $F \subseteq Q$.
Define the following probabilistic tree automata:
\begin{itemize}
	\item $\A_1 = (Q,q_\init,\delta')$ where 
$\delta'(p, a) = \frac{1}{2} (q_1, q_1) + \frac{1}{2} (q_2, q_2)$
if $\delta(p, a) = \frac{1}{2} q_1 + \frac{1}{2} q_2$.
	\item $\A_2 = (Q,q_\init,\delta'')$ where 
$\delta''(p, a) = \frac{1}{2} (q_1, q_2) + \frac{1}{2} (q_2, q_1)$
if $\delta(p, a) = \frac{1}{2} q_1 + \frac{1}{2} q_2$.
\end{itemize} 

{Lemma~\ref{lem:word_to_tree} relates $\A_1$ and $\A$, where $\mu$ denotes the
coin-flipping measure on branches defined in
Section~\ref{sec:univ-tree-automata}, while Lemma~\ref{lem:cross} relates $\A_1$ and $\A_2$.}

\begin{lemma}[{\cite[Proposition~43]{CHS14}}]\label{lem:word_to_tree}
The following holds:
\[
L^{\forall^{=1}}_{\text{Qual},\CoBuchi(F)}(\A_1)=\set{t \in \Tree(\Sigma) : \mu \left( \set{ b \in \set{0,1}^\omega :\ t[b] \in L^{=1}_{\CoBuchi(F)}(\A) } \right) = 1}.
\]
\end{lemma}

{Now, for a fixed tree $t$, 
the Markov chains $\M^t_{\A_1}$ and
$\M^t_{\A_2}$ associated with $\A_1$ and $\A_2$ respectively are equal: indeed, they
have the same states $Q\times\{0,1\}^*$, the same initial state
$(q_\init,\epsilon)$ and the same probability transition function $T$ given
by
$$T((q,u))=\frac{1}{4} (q_1,u0) + \frac{1}{4} (q_1,u1) + \frac{1}{4}
(q_2,u0) + \frac{1}{4} (q_2,u1)$$ where
$\delta(q,t(u))=\frac{1}{2}q_1+\frac{1}{2}q_2$ in $\A$. As a consequence, $\A_1$ and $\A_2$ have the same qualitative co-Büchi  semantics.

\begin{lemma}\label{lem:cross} 
\[
L^{\forall^{=1}}_{\text{Qual},\CoBuchi(F)}(\A_2) = 
L^{\forall^{=1}}_{\text{Qual},\CoBuchi(F)}(\A_1). 
\]
\end{lemma}
}

We are now ready to prove Theorem~\ref{thm:undecidability-univcoB}.

\begin{proof}[Proof of Theorem~\ref{thm:undecidability-univcoB}]
Let $\A = (Q,q_\init,\delta)$ be a simple probabilistic word automaton and $F \subseteq Q$.
{We define the tree automaton $\A_U = (Q,q_\init,\Delta)$ where 
\[
\Delta = \set{(q, a, q_1, q_2 ), (q, a, q_2, q_1) \mid \delta(q, a) = \frac{1}{2} q_1 + \frac{1}{2} q_2}.
\]

Now, we establish that $L^{=1}_{\CoBuchi(F)}(\A) \neq \emptyset$ if, and only if, $L^{\forall}_{\text{Qual},\CoBuchi(F)}(\A_U)\neq \emptyset$, which concludes the proof of Theorem~\ref{thm:undecidability-univcoB}.

Assume that there is some $w=w_0w_1\cdots \in L^{=1}_{\CoBuchi(F)}(\A)$, that is such that $P_\A^w(\CoBuchi(F)) = 1$. We construct  a tree $t_w$ whose branches are all equal to $w$, i.e. $t_w(u)=w_{|u|}$ for every $u\in\{0,1\}^*$.

For a fixed run $\rho$ of $\A_U$ over $t_w$, there is a bijection between the infinite paths
of $\M_{\A}^w$ and $\M_{\A_U}^\rho$ that preserves the measure (it suffices to
notice that the measure is preserved for cones) and also the property of visiting finitely many states in $F$.  As a result, 
$P_\A^w(\CoBuchi(F)) = 1$ entails
$P_{\A_U}^\rho(\CoBuchi(F)) = 1$, for all runs $\rho$. Thus
$t_w \in L^{\forall}_{\text{Qual},\CoBuchi(F)}(\A_U)$.}

The converse implication is not immediate  {because a  tree 
$t \in L^{\forall}_{\text{Qual},\CoBuchi(F)}(\A_U)$ may not necessarily be of the form $t_w$ for some word $w \in \Sigma^\omega$.}

In Section~\ref{def:prob_tree_automata}, we informally said that an equivalent definition of almost-sure acceptance for probabilistic tree automata can be obtained by associating a probability measure on the set of all runs induced by a tree, and {by} requiring the measure of the set of qualitatively accepting runs to be {equal to} $1$; in this approach the notion of a run is the same as for (non-probabilistic) tree automata (see \cite{CHS14} for details).

Now, consider the probabilistic tree automaton $\A_2$ used in Lemma~\ref{lem:cross}: for a fixed tree $t$, the set of runs of $\A_U$ over $t$ is the same as the set of runs of $\A_2$ over $t$.
Since \emph{all} runs of $\A_U$ over $t$ are qualitatively accepted, 
then all runs of $\A_2$ over $t$ are qualitatively accepted {too}, 
so the set of qualitatively accepting runs of $\A_2$ over $t$ has measure $1$. {In other words,} $t \in L^{\forall^{=1}}_{\text{Qual},\CoBuchi(F)}(\A_2)$.
{Hence, by Lemma~\ref{lem:cross}, $t \in L^{\forall^{=1}}_{\text{Qual},\CoBuchi(F)}(\A_1)$}.
Finally, using Lemma~\ref{lem:word_to_tree}, {almost all branches of $t$ are in $L^{=1}_{\CoBuchi(F)}(\A)$, entailing 
$L^{=1}_{\CoBuchi(F)}(\A) \neq\emptyset$.}
\end{proof}

\section{Corollaries}
\label{sec-cor}
In this section we derive {two corollaries from Theorem~\ref{thm:undecidability-univcoB}:
the undecidability of the \MSOzeropath theory of the infinite binary tree (Theorem~\ref{thm:undecidabiltyMSOzeropath}), and the undecidability of the emptiness problem for alternating probabilistic automata with non-zero semantics (Theorem~\ref{thm:alternating}).} 

\subsection{Undecidability of \MSOzeropath} \label{sec:msozero}
%-----------------------------------------------------------

{Before stating the problem, we refer the reader to \cite{Thomas97} for definitions and basic properties on Monadic Second Order logic (\MSO) for trees.

The logic \MSOzero, introduced and studied in~\cite{michalewski2016measure,michalewski2018monadic}, 
extends  \MSO with a probabilistic operator $\forall^{=1} X. \varphi$ stating that the set of all sets $X$ satisfying $\varphi$ contains a subset of Lebesgue-measure one.}
Michalewski, Mio and Skrzypczak proved in these papers that the \MSOzero-theory of the infinite
binary tree is undecidable. 
They also considered a variant of this logic, denoted by \MSOzeropath, 
in which the quantification in the probabilistic operator is restricted to sets of nodes that form a path. 
They proved that, in terms of expressiveness, \MSOzeropath is between  \MSO and \MSOzero, with a strict gain in expressiveness compared to \MSO.
However, they left open the question of the decidability of the \MSOzeropath theory of the infinite binary tree~\cite[Problem 4]{michalewski2018monadic}.

In this section, we establish that in fact \MSOzeropath is undecidable over the infinite binary tree, as a direct consequence of Theorem~\ref{thm:undecidability-univcoB}.

The syntax of \MSOzeropath is given by the following grammar:
\[\varphi ::= \suc_0(x,y) \pipe \suc_1(x,y) \pipe x\in X 
\pipe \neg \varphi \pipe \varphi_1 \wedge \varphi_2 \pipe \forall x.\varphi \pipe 
\forall X.\varphi \pipe \forallonep X.\varphi\]
where $x$ ranges over a countable set of \emph{first-order variables}, and $X$ ranges over a countable set of \emph{monadic second-order variables} (also called \emph{set variables}). The quantifier $\forallonep$ is called the \concept{path-measure quantifier}.

The semantics of \MSO on the infinite binary tree is defined by
interpreting the first-order variables $x$ as nodes, and the set
variables $X$ as subsets of nodes. Ordinary quantification and the
Boolean operations are defined as usual, $x \in X$ is interpreted as
the membership relation, and $\suc_i$ (for $i = 0,1$) is interpreted
as the binary relation $\{(x,x\cdot i)\pipe x\in \words\}$.  We now
describe how to interpret quantified formulas of the form $\forallonep X. \varphi$.
A \concept{path} is a prefix-closed non-empty set
$X \subseteq \{0,1\}^{*}$ such that for any node $v \in X$ either
$v  0 \in X$ or $v 1 \in X$, but not both.
We let $\Paths$ denote the set of all paths. 
Note that there is a one-to-one correspondence between $\Paths$ and the set $\{0,1\}^\omega$ of branches. 
Thus, the coin-flipping measure $\mu$, defined over $\{0,1\}^\omega$ (see Section~\ref{sec:univ-tree-automata}), induces a
measure over $\Paths$, which we write $\mu$.  
We let $t \models \forallonep X. \varphi$ if there exists a measurable subset of paths $\Pi\subseteq \Paths$ with $\mu(\Pi)=1$ 
and such that for any $\pi\in \Pi$ one has $t,\pi \models \varphi$.

A \emph{sentence} is a formula without free variables. 
The \concept{\MSOzeropath-theory} of the infinite binary tree is the set of all \MSOzeropath-sentences $\varphi$ that hold in
the infinite binary tree.

We identify a $\{0,1\}^n$-tree $t$ with a tuple of $n$ subsets of nodes $\mathrm{setTuple}(t)=(X_1,\dots, X_n)$ where a node $x$ belongs to $X_i$ if and only if the $i$-th element of $t(x)$ is $1$. This immediately permits to interpret an \MSOzeropath formula with $n$ free set variables on $\{0,1\}^n$-trees.

The following result is an easy consequence of Theorem~\ref{thm:undecidability-univcoB}.

\begin{theorem}
\label{thm:undecidabiltyMSOzeropath}
  The \MSOzeropath-theory of the infinite binary tree is undecidable.
\end{theorem}

\begin{proof}
We reduce the emptiness problem for co-Büchi tree automata with qualitative semantics, 
that we proved undecidable (Theorem~\ref{thm:undecidability-univcoB}), 
to the \MSOzeropath-theory of the infinite binary tree.

Let $\A$ be a co-Büchi tree automaton over the alphabet $\Sigma$.
Without loss of generality, we assume that $\Sigma \subseteq \set{0,1}^n$ for some $n$. Note that, as \MSOzeropath-formulas are interpreted over the (unlabelled) infinite binary tree, we use tuples of subsets of nodes to encode $\Sigma$-trees. 
We construct an \MSOzeropath {formula} $\varphi(\vec{X})$, with $\vec{X}=(X_1,\dots,X_n)$, such that 
\[
L^{\forall}_{\text{Qual},\CoBuchi(F)}(\A) = \set{t \mid \mathrm{setTuple}(t) \models \varphi }.
\]

The formula $\varphi$ mimics the definition of $L^{\forall}_{\text{Qual},\CoBuchi(F)}(\A)$:
\[
\forall \vec Y. (\text{``$\vec Y$ is a run of $\A$ on $\vec X$''} \Rightarrow 
\forallonep Z.(\text{``$Z$ is an accepting path of $\vec Y$''})),
\]
where 
``$\vec Y$ is a run of $\A$ on $\vec X$'' and 
``$Z$ is an accepting path of $\vec Y$'' are expressed in first-order logic (we refer to~\cite{Rabin:TAMS69} for this classical encoding). The desired formula is then $\neg\exists \vec X \varphi$, {which achieves the proof.}
\end{proof}

\subsection{Undecidability of the Emptiness Problem for Alternating Tree Automata with Non-zero Semantics}
%-------------------------------------------------
\label{sec:nonz}

The non-zero semantics {for tree automata} was introduced by Boja{\'n}czyk, Gimbert and
Kelmendi~\cite{BGK17}.  In a recent paper, Fournier and Gimbert
initiated the study of alternating {tree} automata with non-zero
semantics~\cite{fournier2018alternating}.  Their main result is the
decidability of the emptiness problem for a subclass of these
automata, called \emph{limited choice for \Abelard}, and this is used to
solve the satisfiability problem of CTL$^*$+pCTL$^*$; however the
decidability of emptiness for the full class of alternating
automata with non-zero semantics was left open.  Since this class
easily subsumes universal tree automata with qualitative semantics,
Theorem~\ref{thm:undecidability-univcoB} directly implies that this
problem is undecidable.

An \concept{alternating non-zero automaton} on alphabet $\Sigma$ is a tuple $$\A = ((\setstates,\prec),\state_\init,\setstates_E,\setstates_A,\trans,F_{\forall},F_1,F_{>0})$$ where $\setstates$ is a finite set of states equipped with a total order $\prec$, $\state_\init\in \setstates$ is the initial state, $(\setstates_E,\setstates_A)$ is a partition of $\setstates$ into \Eloise's and \Abelard's states, $\trans$ is a set  of transitions made of \emph{local transitions} (elements of $\setstates\times\Sigma\times \setstates$) and \emph{split transitions} (elements of $\setstates\times\Sigma\times \setstates\times\setstates $), and $F_{\forall},F_1,F_{>0}\subseteq Q$ are subsets of $Q$ defining the semantics of the acceptance game {$\arena_{\A,t}^{\text{n-z}}$, to be defined later}.

The input of such an automaton is a $\Sigma$-tree $\tree$ and
acceptance is defined thanks to a two-player perfect-information stochastic game. The arena is quite similar to the arena $\arena_{\A,t}^{=1}$ defined in Section~\ref{sec:alt-tree-automata} for alternating tree automata with qualitative semantics (simply ignore the total order $\prec$ and subsets $F_{\forall},F_1,F_{>0}$), except that local transitions are handled without interacting with the Random player (i.e. when \Eloise or \Abelard simulates a local transition the state is simply updated and the pebble stays in the same node). 

Formally one lets $G = (V_E \cup V_A \cup V_R,E)$ 
with $V_E = Q_E\times\{0,1\}^*$, $V_A =Q_A\times \{0,1\}^*$ 
and $V_R = \{(q,u,q_0,q_1) \mid u \in \{0,1\}^* \text{ and } (q,t(u),q_0,q_1) \in \Delta\}$, and
$$\begin{array}{ll}
E \qquad = \qquad & \{((q,u),(q',u)) \mid (q,t(u),q') \in \Delta)\} \quad \cup \\ 
 & \{((q,u),(q,u,q_0,q_1)) \mid (q,t(u),q_0,q_1) \in \Delta)\} \quad \cup \\ 
&\{((q,u,q_0,q_1),(q_x,u \cdot x)) \mid x \in \{0,1\} \text{ and } (q,u,q_0,q_1) \in V_R)\}\ 
\end{array}$$

Then, we define $\arena_{\A,t}^{\text{n-z}} = (G,\VE,\VA,\VR,\delta,(q_\init,\epsilon))$ where $\delta((q,u,q_0,q_1))=\frac{1}{2}(q_0,u0) + \frac{1}{2}(q_1,u1)$.

A strategy $\sigma_E$ for {\Eloise} \concept{beats} a strategy for \Abelard $\sigma_A$ if all the following conditions are satisfied:
\begin{itemize}
	\item[(i)] \emph{Sure winning}: in every play consistent with $(\sigma_E,\sigma_A)$ the largest (with respect to $\prec$) state appearing infinitely often belongs to $F_\forall$.
	\item[(ii)] \emph{Almost-sure winning}: the (measurable) set of plays consistent with $(\sigma_E,\sigma_A)$ where the largest state (with respect to $\prec$) appearing infinitely often belongs to $F_1$ has measure $1$.
	\item[(iii)] \emph{Positive winning}: for every history consistent with $(\sigma_E,\sigma_A)$ 
	that ends with a state in $F_{>0}$, the (measurable) set of infinite continuations of this history that contain only states in $F_{> 0}$ and are consistent with $(\sigma_E,\sigma_A)$, has non-zero measure.
\end{itemize}

Finally, a tree $t$ is \concept{accepted} by $\A$ if, and only if, \Eloise has a strategy that beats any strategy of \Abelard. The emptiness problem asks for a given alternating non-zero automaton whether the set of accepted trees is empty.

It is easily seen that alternating automata with non-zero semantics subsume universal co-Büchi tree automata with qualitative semantics. Indeed, consider a universal co-Büchi tree automaton $\A$ with qualitative semantics having a set of states $Q$ and a set of states $F\subseteq Q$  defining the co-Büchi condition. Then, universality is captured by alternation (see Proposition~\ref{proposition:altvsuniv}) and the co-Büchi qualitative acceptance condition of $\A$ can be expressed by part (ii) of the beating condition: it is enough to rank the  states in $F$ higher than those in $Q \setminus F$ in the total order on $Q$, and to let $F_1=Q \setminus F$. 

Together with Theorem~\ref{thm:undecidability-univcoB} this yields the following undecidability result.

\begin{theorem}\label{thm:alternating}
The emptiness problem for alternating tree automata with non-zero semantics is undecidable.
\end{theorem}

\begin{proof}
Consider a co-Büchi universal tree automaton $\A=(\setstates, q_\init, \trans)$
whose acceptance condition is given by a subset
$F\subseteq \setstates$.  Without loss of generality, we can safely assume
that $\setstates = \{q_1,\dots q_n\}$ where $n=|Q|$ and that
$F=\{q_k,\dots,q_n\}$ for some $k\leq n+1$.  We construct an
alternating non-zero automaton
${\B} =
((\setstates,\prec),q_\init,\emptyset,\setstates,\trans,\setstates,F_1,\setstates)$,
where the total order $\prec$ on $Q$ is defined by $q_i \prec q_j$ if
and only if $i<j$ and  $F_1=Q\setminus F$. 

Note that {since $\A$ has only split transitions, the arenas $\arena_{\A,t}^{=1}$ and $\arena_{\B,t}^{\text{n-z}}$ are the same  for any tree $t$, and so are the strategies for \Eloise and \Abelard.} Moreover, it is immediate that an \Eloise's strategy $\sigma_E$ beats an \Abelard's strategy $\sigma_A$ in $\arena_{\B,t}^{\text{n-z}}$ if, and only if, almost all plays in $\arena_{\A,t}^{=1}$ consistent with $(\sigma_E,\sigma_A)$ {satisfy} the co-Büchi condition. Hence, \Eloise has a strategy that beats any strategy of \Abelard in $\arena_{\B,t}^{\text{n-z}}$ if and only if she has an almost-surely winning strategy in the co-Büchi game $(\arena_{\A,t}^{=1},F\times\{0,1\}^*)$. {Otherwise said, a tree is accepted by $\B$ if, and only if, it is accepted by $\A$.}

{Applying Theorem~\ref{thm:undecidability-univcoB},  concludes the proof.}
\end{proof}

\section*{Conclusions}
\label{sec:conclusions}

{The core contribution is the study of alternating automata with qualitative semantics and the identification of a sharp decidability frontier for their emptiness problem: the emptiness problem is decidable for B{\"u}chi objectives, but it is undecidable for the co-B{\"u}chi objectives. The latter  undecidability result directly implies the undecidability of \MSOzeropath in an elegant manner.
In an attempt to exhibiting a decidable extension of \MSO with a probabilistic operator, a natural track is to seek natural subclasses of alternating tree automata with qualitative semantics (or even of non-zero automata) with a decidable emptiness problem. However, while for alternating Buchi tree automata with qualitative semantics emptiness problem is decidable, their connection with a robust logic is unclear. The recent results concerning restrictions to thin quantification~\cite{Bojanczyk16,BGK17} and to limited choice for \Abelard~\cite{fournier2018alternating} bring hope and inspiration for the construction of such subclasses.}

\bibliographystyle{ACM-Reference-Format}
\bibliography{main}


\begin{thebibliography}{33}

%%% the Web address.

\ifx \showCODEN    \undefined \def \showCODEN     #1{\unskip}     \fi
\ifx \showDOI      \undefined \def \showDOI       #1{#1}\fi
\ifx \showISBNx    \undefined \def \showISBNx     #1{\unskip}     \fi
\ifx \showISBNxiii \undefined \def \showISBNxiii  #1{\unskip}     \fi
\ifx \showISSN     \undefined \def \showISSN      #1{\unskip}     \fi
\ifx \showLCCN     \undefined \def \showLCCN      #1{\unskip}     \fi
\ifx \shownote     \undefined \def \shownote      #1{#1}          \fi
\ifx \showarticletitle \undefined \def \showarticletitle #1{#1}   \fi
\ifx \showURL      \undefined \def \showURL       {\relax}        \fi
% invisible to TeX
\providecommand\bibfield[2]{#2}
\providecommand\bibinfo[2]{#2}
\providecommand\natexlab[1]{#1}
\providecommand\showeprint[2][]{arXiv:#2}

\bibitem[\protect\citeauthoryear{Baier, Gr{\"o}{\ss}er, and Bertrand}{Baier
  et~al\mbox{.}}{2012}]%
        {BGB12}
\bibfield{author}{\bibinfo{person}{Christel Baier}, \bibinfo{person}{Marcus
  Gr{\"o}{\ss}er}, {and} \bibinfo{person}{Nathalie Bertrand}.}
  \bibinfo{year}{2012}\natexlab{}.
\newblock \showarticletitle{Probabilistic $\omega$-Automata}.
\newblock \bibinfo{journal}{\emph{J. ACM}} \bibinfo{volume}{59},
  \bibinfo{number}{1} (\bibinfo{year}{2012}), \bibinfo{pages}{1}.
\newblock


\bibitem[\protect\citeauthoryear{B{\'a}r{\'a}ny, Kaiser, and
  Rabinovich}{B{\'a}r{\'a}ny et~al\mbox{.}}{2010}]%
        {Rabinovich:FI10}
\bibfield{author}{\bibinfo{person}{Vince B{\'a}r{\'a}ny},
  \bibinfo{person}{{\L}ukasz Kaiser}, {and} \bibinfo{person}{Alex Rabinovich}.}
  \bibinfo{year}{2010}\natexlab{}.
\newblock \showarticletitle{Expressing Cardinality Quantifiers in Monadic
  Second-Order Logic over Trees}.
\newblock \bibinfo{journal}{\emph{Fundamenta Informaticae}}
  \bibinfo{volume}{100}, \bibinfo{number}{1-4} (\bibinfo{year}{2010}),
  \bibinfo{pages}{1--17}.
\newblock


\bibitem[\protect\citeauthoryear{Berthon, Filiot, Guha, Maubert, Murano,
  Pinault, Raskin, and Rubin}{Berthon et~al\mbox{.}}{2019}]%
        {Berthonetal19}
\bibfield{author}{\bibinfo{person}{Raphaël Berthon}, \bibinfo{person}{Emmanuel
  Filiot}, \bibinfo{person}{Shibashis Guha}, \bibinfo{person}{Bastien Maubert},
  \bibinfo{person}{Nello Murano}, \bibinfo{person}{Laureline Pinault},
  \bibinfo{person}{Jean-François Raskin}, {and} \bibinfo{person}{Sasha
  Rubin}.} \bibinfo{year}{2019}\natexlab{}.
\newblock \showarticletitle{Monadic Second-Order Logic with Path-Measure
  Quantifier is Undecidable}.
\newblock \bibinfo{journal}{\emph{CoRR}}  \bibinfo{volume}{abs/1901.04349}
  (\bibinfo{year}{2019}).
\newblock


\bibitem[\protect\citeauthoryear{Bertrand, Genest, and Gimbert}{Bertrand
  et~al\mbox{.}}{2009}]%
        {BGG09}
\bibfield{author}{\bibinfo{person}{Nathalie Bertrand}, \bibinfo{person}{Blaise
  Genest}, {and} \bibinfo{person}{Hugo Gimbert}.}
  \bibinfo{year}{2009}\natexlab{}.
\newblock \showarticletitle{Qualitative Determinacy and Decidability of
  Stochastic Games with Signals}. In \bibinfo{booktitle}{\emph{Proceedings of
  the 24th Annual IEEE Symposium on Logic in Computer Science}}.
  \bibinfo{publisher}{IEEE}, \bibinfo{pages}{319--328}.
\newblock


\bibitem[\protect\citeauthoryear{Boja{\'n}czyk}{Boja{\'n}czyk}{2016}]%
        {Bojanczyk16}
\bibfield{author}{\bibinfo{person}{Miko{\l}aj Boja{\'n}czyk}.}
  \bibinfo{year}{2016}\natexlab{}.
\newblock \showarticletitle{Thin {MSO} with a Probabilistic Path Quantifier}.
  In \bibinfo{booktitle}{\emph{Proceedings of the 43rd International Colloquium
  on Automata, Languages, and Programming}} \emph{(\bibinfo{series}{LIPIcs},
  Vol.~\bibinfo{volume}{55})}. \bibinfo{publisher}{Schloss Dagstuhl -
  Leibniz-Zentrum fuer Informatik}, \bibinfo{pages}{96:1--96:13}.
\newblock


\bibitem[\protect\citeauthoryear{Boja{\'n}czyk, Gimbert, and
  Kelmendi}{Boja{\'n}czyk et~al\mbox{.}}{2017}]%
        {BGK17}
\bibfield{author}{\bibinfo{person}{Miko{\l}aj Boja{\'n}czyk},
  \bibinfo{person}{Hugo Gimbert}, {and} \bibinfo{person}{Edon Kelmendi}.}
  \bibinfo{year}{2017}\natexlab{}.
\newblock \showarticletitle{Emptiness of Zero Automata Is Decidable}. In
  \bibinfo{booktitle}{\emph{Proceedings of the 44th International Colloquium on
  Automata, Languages, and Programming}} \emph{(\bibinfo{series}{LIPIcs},
  Vol.~\bibinfo{volume}{80})}. \bibinfo{publisher}{Schloss Dagstuhl -
  Leibniz-Zentrum fuer Informatik}, \bibinfo{pages}{106:1--106:13}.
\newblock


\bibitem[\protect\citeauthoryear{Boja{\'{n}}czyk, Kelmendi, and {Micha{l}
  Skrzypczak}}{Boja{\'{n}}czyk et~al\mbox{.}}{2019}]%
        {BKS19}
\bibfield{author}{\bibinfo{person}{Miko{\l}aj Boja{\'{n}}czyk},
  \bibinfo{person}{Edon Kelmendi}, {and} \bibinfo{person}{{Micha{l}
  Skrzypczak}}.} \bibinfo{year}{2019}\natexlab{}.
\newblock \showarticletitle{MSO+{\(\nabla\)} is Undecidable}. In
  \bibinfo{booktitle}{\emph{Proceedings of the 34th Annual {ACM/IEEE} Symposium
  on Logic in Computer Science}}. \bibinfo{publisher}{IEEE},
  \bibinfo{pages}{1--13}.
\newblock


\bibitem[\protect\citeauthoryear{Carayol, Haddad, and Serre}{Carayol
  et~al\mbox{.}}{2014}]%
        {CHS14}
\bibfield{author}{\bibinfo{person}{Arnaud Carayol}, \bibinfo{person}{Axel
  Haddad}, {and} \bibinfo{person}{Olivier Serre}.}
  \bibinfo{year}{2014}\natexlab{}.
\newblock \showarticletitle{Randomization in Automata on Infinite Trees}.
\newblock \bibinfo{journal}{\emph{{ACM} Transactions on Computational Logic}}
  \bibinfo{volume}{15}, \bibinfo{number}{3} (\bibinfo{year}{2014}),
  \bibinfo{pages}{24:1--24:33}.
\newblock


\bibitem[\protect\citeauthoryear{Carayol, L{\"{o}}ding, and Serre}{Carayol
  et~al\mbox{.}}{2018}]%
        {CarayolLS18}
\bibfield{author}{\bibinfo{person}{Arnaud Carayol}, \bibinfo{person}{Christof
  L{\"{o}}ding}, {and} \bibinfo{person}{Olivier Serre}.}
  \bibinfo{year}{2018}\natexlab{}.
\newblock \showarticletitle{Pure Strategies in Imperfect Information Stochastic
  Games}.
\newblock \bibinfo{journal}{\emph{Fundamenta Informaticae}}
  \bibinfo{volume}{160}, \bibinfo{number}{4} (\bibinfo{year}{2018}),
  \bibinfo{pages}{361--384}.
\newblock


\bibitem[\protect\citeauthoryear{Chatterjee and Doyen}{Chatterjee and
  Doyen}{2014}]%
        {DBLP:journals/tocl/Chatterjee014}
\bibfield{author}{\bibinfo{person}{Krishnendu Chatterjee} {and}
  \bibinfo{person}{Laurent Doyen}.} \bibinfo{year}{2014}\natexlab{}.
\newblock \showarticletitle{Partial-Observation Stochastic Games: How to Win
  when Belief Fails}.
\newblock \bibinfo{journal}{\emph{{ACM} Transactions on Computational Logic}}
  \bibinfo{volume}{15}, \bibinfo{number}{2} (\bibinfo{year}{2014}),
  \bibinfo{pages}{16:1--16:44}.
\newblock


\bibitem[\protect\citeauthoryear{Chatterjee, Doyen, Henzinger, and
  Raskin}{Chatterjee et~al\mbox{.}}{2007}]%
        {CDHR07}
\bibfield{author}{\bibinfo{person}{Krishnendu Chatterjee},
  \bibinfo{person}{Laurent Doyen}, \bibinfo{person}{Thomas~A. Henzinger}, {and}
  \bibinfo{person}{Jean-Fran\c{c}ois Raskin}.} \bibinfo{year}{2007}\natexlab{}.
\newblock \showarticletitle{Algorithms for Omega-Regular Games with Imperfect
  Information}.
\newblock \bibinfo{journal}{\emph{Logical Methods in Computer Science}}
  \bibinfo{volume}{3}, \bibinfo{number}{3} (\bibinfo{year}{2007}).
\newblock


\bibitem[\protect\citeauthoryear{Courcoubetis and Yannakakis}{Courcoubetis and
  Yannakakis}{1990}]%
        {CourcoubetisY90}
\bibfield{author}{\bibinfo{person}{Costas Courcoubetis} {and}
  \bibinfo{person}{Mihalis Yannakakis}.} \bibinfo{year}{1990}\natexlab{}.
\newblock \showarticletitle{Markov Decision Processes and Regular Events
  (Extended Abstract)}. In \bibinfo{booktitle}{\emph{Proceedings of the 17th
  International Colloquium on Automata, Languages, and Programming (ICALP
  1990)}} \emph{(\bibinfo{series}{Lecture Notes in Computer Science},
  Vol.~\bibinfo{volume}{443})}. \bibinfo{publisher}{Springer},
  \bibinfo{pages}{336--349}.
\newblock


\bibitem[\protect\citeauthoryear{de~Alfaro}{de~Alfaro}{1999}]%
        {DeAlfaro99}
\bibfield{author}{\bibinfo{person}{Luca de Alfaro}.}
  \bibinfo{year}{1999}\natexlab{}.
\newblock \showarticletitle{The Verification of Probabilistic Systems Under
  Memoryless Partial-Information Policies is Hard}. In
  \bibinfo{booktitle}{\emph{Proceedings of the 2nd International Workshop on
  Probabilistic Methods in Verification}}. \bibinfo{pages}{19--32}.
\newblock


\bibitem[\protect\citeauthoryear{Fagin, Halpern, Moses, and Vardi}{Fagin
  et~al\mbox{.}}{1995}]%
        {FaginHMV95}
\bibfield{author}{\bibinfo{person}{Ronald Fagin}, \bibinfo{person}{Joseph~Y.
  Halpern}, \bibinfo{person}{Yoram. Moses}, {and} \bibinfo{person}{Moshe~Y.
  Vardi}.} \bibinfo{year}{1995}\natexlab{}.
\newblock \bibinfo{booktitle}{\emph{Reasoning about Knowledge}}.
\newblock \bibinfo{publisher}{MIT Press}.
\newblock


\bibitem[\protect\citeauthoryear{Fijalkow}{Fijalkow}{2017}]%
        {Fijalkow17}
\bibfield{author}{\bibinfo{person}{Nathana{\"{e}}l Fijalkow}.}
  \bibinfo{year}{2017}\natexlab{}.
\newblock \showarticletitle{Undecidability Results for Probabilistic Automata}.
\newblock \bibinfo{journal}{\emph{{SIGLOG} News}} \bibinfo{volume}{4},
  \bibinfo{number}{4} (\bibinfo{year}{2017}), \bibinfo{pages}{10--17}.
\newblock


\bibitem[\protect\citeauthoryear{Fijalkow, Gimbert, Kelmendi, and
  Oualhadj}{Fijalkow et~al\mbox{.}}{2015}]%
        {FGKO15}
\bibfield{author}{\bibinfo{person}{Nathana{\"{e}}l Fijalkow},
  \bibinfo{person}{Hugo Gimbert}, \bibinfo{person}{Edon Kelmendi}, {and}
  \bibinfo{person}{Youssouf Oualhadj}.} \bibinfo{year}{2015}\natexlab{}.
\newblock \showarticletitle{Deciding the Value 1 Problem for Probabilistic
  Leaktight Automata}.
\newblock \bibinfo{journal}{\emph{Logical Methods in Computer Science}}
  \bibinfo{volume}{11}, \bibinfo{number}{2} (\bibinfo{year}{2015}),
  \bibinfo{pages}{1--42}.
\newblock


\bibitem[\protect\citeauthoryear{Fijalkow, Pinchinat, and Serre}{Fijalkow
  et~al\mbox{.}}{2013a}]%
        {fijalkow2013emptiness}
\bibfield{author}{\bibinfo{person}{Nathana{\"e}l Fijalkow},
  \bibinfo{person}{Sophie Pinchinat}, {and} \bibinfo{person}{Olivier Serre}.}
  \bibinfo{year}{2013}\natexlab{a}.
\newblock \showarticletitle{Emptiness of Alternating Tree Automata Using Games
  with Imperfect Information}. In \bibinfo{booktitle}{\emph{Proceedings of
  {IARCS} Annual Conference on Foundations of Software Technology and
  Theoretical Computer Science}} \emph{(\bibinfo{series}{LIPIcs},
  Vol.~\bibinfo{volume}{24})}. \bibinfo{publisher}{Schloss Dagstuhl -
  Leibniz-Zentrum fuer Informatik}, \bibinfo{pages}{299--311}.
\newblock


\bibitem[\protect\citeauthoryear{Fijalkow, Pinchinat, and Serre}{Fijalkow
  et~al\mbox{.}}{2013b}]%
        {fijalkow2013emptinessComplete}
\bibfield{author}{\bibinfo{person}{Nathana{\"e}l Fijalkow},
  \bibinfo{person}{Sophie Pinchinat}, {and} \bibinfo{person}{Olivier Serre}.}
  \bibinfo{year}{2013}\natexlab{b}.
\newblock \bibinfo{title}{Emptiness of Alternating Tree Automata Using Games
  with Imperfect Information}.  (\bibinfo{year}{2013}).
\newblock
\urldef\tempurl%
\url{https://hal.inria.fr/hal-01260682}
\showURL{%
\tempurl}


\bibitem[\protect\citeauthoryear{Fournier and Gimbert}{Fournier and
  Gimbert}{2018}]%
        {fournier2018alternating}
\bibfield{author}{\bibinfo{person}{Paulin Fournier} {and} \bibinfo{person}{Hugo
  Gimbert}.} \bibinfo{year}{2018}\natexlab{}.
\newblock \showarticletitle{Alternating Nonzero Automata}. In
  \bibinfo{booktitle}{\emph{Proceedings of the 29th International Conference on
  Concurrency Theory}} \emph{(\bibinfo{series}{LIPIcs},
  Vol.~\bibinfo{volume}{118})}. \bibinfo{publisher}{Schloss Dagstuhl -
  Leibniz-Zentrum fuer Informatik}, \bibinfo{pages}{13:1--13:16}.
\newblock


\bibitem[\protect\citeauthoryear{Gimbert and Oualhadj}{Gimbert and
  Oualhadj}{2010}]%
        {GO10}
\bibfield{author}{\bibinfo{person}{Hugo Gimbert} {and}
  \bibinfo{person}{Youssouf Oualhadj}.} \bibinfo{year}{2010}\natexlab{}.
\newblock \showarticletitle{Probabilistic Automata on Finite Words: Decidable
  and Undecidable Problems}. In \bibinfo{booktitle}{\emph{Proceedings of the
  37th International Colloquium on Automata, Languages and Programming}}
  \emph{(\bibinfo{series}{Lecture Notes in Computer Science},
  Vol.~\bibinfo{volume}{6199})}. \bibinfo{publisher}{Springer},
  \bibinfo{pages}{527--538}.
\newblock


\bibitem[\protect\citeauthoryear{Gimbert and Zielonka}{Gimbert and
  Zielonka}{2007}]%
        {GZ07}
\bibfield{author}{\bibinfo{person}{Hugo Gimbert} {and}
  \bibinfo{person}{Wies{\l}aw Zielonka}.} \bibinfo{year}{2007}\natexlab{}.
\newblock \showarticletitle{Perfect Information Stochastic Priority Games}. In
  \bibinfo{booktitle}{\emph{Proceedings of the 34th International Colloquium on
  Automata, Languages, and Programming}} \emph{(\bibinfo{series}{Lecture Notes
  in Computer Science}, Vol.~\bibinfo{volume}{4596})}.
  \bibinfo{publisher}{Springer}, \bibinfo{pages}{850--861}.
\newblock


\bibitem[\protect\citeauthoryear{Gripon and Serre}{Gripon and Serre}{2009}]%
        {GS09}
\bibfield{author}{\bibinfo{person}{Vincent Gripon} {and}
  \bibinfo{person}{Olivier Serre}.} \bibinfo{year}{2009}\natexlab{}.
\newblock \showarticletitle{Qualitative Concurrent Stochastic Games with
  Imperfect Information}. In \bibinfo{booktitle}{\emph{Proceedings of the 36th
  International Colloquium on Automata, Languages, and Programming}}
  \emph{(\bibinfo{series}{Lecture Notes in Computer Science},
  Vol.~\bibinfo{volume}{5556})}. \bibinfo{publisher}{Springer},
  \bibinfo{pages}{200--211}.
\newblock


\bibitem[\protect\citeauthoryear{Ku\v{c}era}{Ku\v{c}era}{2011}]%
        {Kucera_foxes}
\bibfield{author}{\bibinfo{person}{Anton\'in Ku\v{c}era}.}
  \bibinfo{year}{2011}\natexlab{}.
\newblock \showarticletitle{Turn-Based Stochastic Games}.
\newblock In \bibinfo{booktitle}{\emph{Lectures in Game Theory for Computer
  Scientists}}, \bibfield{editor}{\bibinfo{person}{Krzysztof~R. Apt} {and}
  \bibinfo{person}{Erich Grdel}} (Eds.). \bibinfo{publisher}{Cambridge
  University Press}, \bibinfo{address}{New York, NY, USA}, Chapter~5,
  \bibinfo{pages}{146--184}.
\newblock


\bibitem[\protect\citeauthoryear{Michalewski and Mio}{Michalewski and
  Mio}{2016}]%
        {michalewski2016measure}
\bibfield{author}{\bibinfo{person}{Henryk Michalewski} {and}
  \bibinfo{person}{Matteo Mio}.} \bibinfo{year}{2016}\natexlab{}.
\newblock \showarticletitle{Measure Quantifier in Monadic Second Order Logic}.
  In \bibinfo{booktitle}{\emph{Proceedings of Logical Foundations of Computer
  Science - International Symposium}} \emph{(\bibinfo{series}{Lecture Notes in
  Computer Science}, Vol.~\bibinfo{volume}{9537})}.
  \bibinfo{publisher}{Springer}, \bibinfo{pages}{267--282}.
\newblock


\bibitem[\protect\citeauthoryear{Mio, Skrzypczak, and Michalewski}{Mio
  et~al\mbox{.}}{2018}]%
        {michalewski2018monadic}
\bibfield{author}{\bibinfo{person}{Matteo Mio}, \bibinfo{person}{Micha{\l}
  Skrzypczak}, {and} \bibinfo{person}{Henryk Michalewski}.}
  \bibinfo{year}{2018}\natexlab{}.
\newblock \showarticletitle{Monadic Second Order Logic with Measure and
  Category Quantifiers}.
\newblock \bibinfo{journal}{\emph{Logical Methods in Computer Science}}
  \bibinfo{volume}{14}, \bibinfo{number}{2} (\bibinfo{year}{2018}).
\newblock


\bibitem[\protect\citeauthoryear{Paz}{Paz}{1971}]%
        {paz2014introduction}
\bibfield{author}{\bibinfo{person}{Azaria Paz}.}
  \bibinfo{year}{1971}\natexlab{}.
\newblock \bibinfo{booktitle}{\emph{Introduction to Probabilistic Automata}}.
\newblock \bibinfo{publisher}{Academic Press}.
\newblock


\bibitem[\protect\citeauthoryear{Puterman}{Puterman}{1994}]%
        {Puterman94}
\bibfield{author}{\bibinfo{person}{Martin~L. Puterman}.}
  \bibinfo{year}{1994}\natexlab{}.
\newblock \bibinfo{booktitle}{\emph{Markov Decision Processes: Discrete
  Stochastic Dynamic Programming}}.
\newblock \bibinfo{publisher}{John Wiley \& Sons, Inc.}, \bibinfo{address}{New
  York, NY, USA}.
\newblock


\bibitem[\protect\citeauthoryear{Rabin}{Rabin}{1963}]%
        {Rabin:IC63}
\bibfield{author}{\bibinfo{person}{Michael~O. Rabin}.}
  \bibinfo{year}{1963}\natexlab{}.
\newblock \showarticletitle{Probabilistic Automata}.
\newblock \bibinfo{journal}{\emph{Information and Control}}
  \bibinfo{volume}{6}, \bibinfo{number}{3} (\bibinfo{year}{1963}),
  \bibinfo{pages}{230--245}.
\newblock


\bibitem[\protect\citeauthoryear{Rabin}{Rabin}{1969}]%
        {Rabin:TAMS69}
\bibfield{author}{\bibinfo{person}{Michael~O. Rabin}.}
  \bibinfo{year}{1969}\natexlab{}.
\newblock \showarticletitle{Decidability of Second-Order Theories and Automata
  on Infinite Trees}.
\newblock \bibinfo{journal}{\emph{Transactions of the AMS}}
  \bibinfo{volume}{141} (\bibinfo{year}{1969}), \bibinfo{pages}{1--35}.
\newblock


\bibitem[\protect\citeauthoryear{Reif}{Reif}{1979}]%
        {Reif79}
\bibfield{author}{\bibinfo{person}{J.H. Reif}.}
  \bibinfo{year}{1979}\natexlab{}.
\newblock \showarticletitle{Universal Games of Incomplete Information}. In
  \bibinfo{booktitle}{\emph{Proc of STOC'79}}. \bibinfo{publisher}{ACM},
  \bibinfo{pages}{288--308}.
\newblock


\bibitem[\protect\citeauthoryear{Reif}{Reif}{1984}]%
        {Reif84}
\bibfield{author}{\bibinfo{person}{J.H. Reif}.}
  \bibinfo{year}{1984}\natexlab{}.
\newblock \showarticletitle{The Complexity of Two-Player Games of Incomplete
  Information}.
\newblock \bibinfo{journal}{\emph{J. Comput. System Sci.}}
  \bibinfo{volume}{29}, \bibinfo{number}{2} (\bibinfo{year}{1984}),
  \bibinfo{pages}{274--301}.
\newblock


\bibitem[\protect\citeauthoryear{Thomas}{Thomas}{1997}]%
        {Thomas97}
\bibfield{author}{\bibinfo{person}{Wolfgang Thomas}.}
  \bibinfo{year}{1997}\natexlab{}.
\newblock \showarticletitle{Languages, Automata, and Logic}.
\newblock In \bibinfo{booktitle}{\emph{Handbook of Formal Language Theory}},
  \bibfield{editor}{\bibinfo{person}{G.~Rozenberg} {and}
  \bibinfo{person}{A.~Salomaa}} (Eds.). Vol.~\bibinfo{volume}{III}.
  \bibinfo{pages}{389--455}.
\newblock


\bibitem[\protect\citeauthoryear{Zielonka}{Zielonka}{1998}]%
        {Zie98}
\bibfield{author}{\bibinfo{person}{Wies{\l}aw Zielonka}.}
  \bibinfo{year}{1998}\natexlab{}.
\newblock \showarticletitle{Infinite Games on Finitely Coloured Graphs with
  Applications to Automata on Infinite Trees}.
\newblock \bibinfo{journal}{\emph{Theoretical Computer Science}}
  \bibinfo{volume}{200}, \bibinfo{number}{1-2} (\bibinfo{year}{1998}),
  \bibinfo{pages}{135--183}.
\newblock


\end{thebibliography}

\end{document}